\documentclass{vldb-arXiv}

\usepackage{graphicx}

\usepackage{times}
\usepackage{amsmath}
\usepackage{amssymb}
\usepackage{subfigure}
\usepackage{mathtools}

\usepackage[T1]{fontenc}
\usepackage{epsfig}
\usepackage{fancyvrb}
\usepackage{color}
\usepackage{stmaryrd}
\usepackage{fancyvrb}
\usepackage{array}

\usepackage{multirow}
\usepackage{multicol}
\usepackage{latexsym}
\usepackage{theorem}
\usepackage{paralist}
\usepackage{calc}
\usepackage{wrapfig}

\usepackage{txfonts}
\usepackage{floatflt}
\usepackage{picins}

\usepackage{graphicx}

\usepackage{ifpdf}

\newif\ifpdf
\ifx\pdfoutput\undefined
  \pdffalse
  \usepackage[colorlinks=false]{hyperref}
\else
  \pdfoutput=1
  \pdftrue
  \usepackage[pdftex,colorlinks=false,urlcolor=blue,linkcolor=blue]{hyperref}
  \usepackage{pslatex}
\fi

\toappear{}

\newcommand{\hide}[1]{}

\newcommand{\minisec}[1]{\vspace*{0.05cm}\noindent\textbf{#1.}}

\theoremstyle{plain}
\theoremheaderfont{\scshape}
\theorembodyfont{\normalfont}

\makeatletter
\newif\if@restonecol
\makeatother

\usepackage[ruled]{algorithm2e}
\usepackage{algorithmic}

\newtheorem{theorem}{Theorem}

\newcommand{\gobble}[5]{#1 #2 #3 #4}
\newcommand{\referencenumber}{1}
\newtheorem{@theoremreference}{Theorem \referencenumber \ignorespaces\gobble}
\newenvironment{theoremreference}[1]
  {\renewcommand{\referencenumber}{#1}\begin{@theoremreference}}
  {\end{@theoremreference}}

\title{Behavioral Simulations in MapReduce}

\numberofauthors{1}
\author {
  \alignauthor Guozhang Wang, Marcos Vaz Salles, Benjamin Sowell, Xun Wang, Tuan Cao,\\
  Alan Demers, Johannes Gehrke,  Walker White \\
  \vspace{0.75ex}
  \affaddr{Cornell University}\\
  \affaddr{Ithaca, NY 14853, USA} \\
  \email{\{guoz, vmarcos, sowell, tuancao, ademers, johannes, wmwhite\}@cs.cornell.edu}\\
  \email{\{xw239\}@cornell.edu}\\
}

\begin{document}

\newcount\hours
\newcount\minutes        
\newcommand{\timeofday}{ 
\hours=\time
\minutes=\hours
\divide\hours by60
\multiply\hours by60
\advance\minutes by-\hours
\divide\hours by60
\ifnum\hours>9\else0\fi\the\hours:\ifnum\minutes>9\else
0\fi\the\minutes}

\maketitle

\newcounter{enum}
\newenvironment{packed_enum}{
\begin{list}{\arabic{enum}.}{
  \setlength{\itemsep}{-3pt}
  \setlength{\parskip}{0pt}
  \setlength{\labelwidth}{-4 pt}
  \setlength{\leftmargin}{0 pt}
  \setlength{\itemindent}{0pt}
  \usecounter{enum}}
}{\end{list}}

\newcommand{\threefigures}[9]{
  \begin{figure*}[tb]
    \begin{minipage}{0.327\linewidth}
      \centerline{\includegraphics[width=0.9\linewidth]{#1}}
      \vspace{-1ex}
      \caption{#2} \label{#3}
    \end{minipage} \hfill
    \begin{minipage}{0.327\linewidth}
      \centerline{\includegraphics[width=0.9\linewidth]{#4}}
      \vspace{-1ex}
      \caption{#5} \label{#6}
    \end{minipage} \hfill
    \begin{minipage}{0.327\linewidth}
      \centerline{\includegraphics[width=0.9\linewidth]{#7}}
      \vspace{-1ex}
      \caption{#8} \label{#9}
    \end{minipage}
    \vspace{-2ex}
  \end{figure*}}

\newcommand{\triptych}[5]{
  \begin{figure*}[tb]
    \begin{minipage}{0.327\linewidth}
      \centerline{\includegraphics[width=0.9\linewidth]{#1}}
      \centerline{(a)}
    \end{minipage} \hfill
    \begin{minipage}{0.327\linewidth}
      \centerline{\includegraphics[width=0.9\linewidth]{#2}}
      \centerline{(b)}
    \end{minipage} \hfill
    \begin{minipage}{0.327\linewidth}
      \centerline{\includegraphics[width=0.9\linewidth]{#3}}
      \centerline{(c)}
    \end{minipage}
    \vspace{-2ex}
    \caption{#2} \label{#3}
  \end{figure*}}

\newcommand{\C}{\mathcal}
\newcommand{\D}{\Bbb}
\newcommand{\F}{\mathfrak}
\newcommand{\T}{\texttt}
\newcommand{\la}{\langle}
\newcommand{\ra}{\rangle}
\newcommand{\E}[1]{\ensuremath{\langle\T{#1}\rangle}}

\newcommand{\etype}{\ensuremath{\,\Vdash \,}}
\newcommand{\type}{\ensuremath{\,\vdash \,}}
\newcommand{\valid}{\ensuremath{\,\vDash \,}}
\newcommand{\tor}{\ensuremath{\>\>\big|\>\>}}

\newcommand{\dexp}{\ensuremath{::=}}

\newcommand{\key}{\textsc{key}}
\newcommand{\nil}{\textsc{nil}}
\newcommand{\id}{\textsc{id}}
\newcommand{\sng}{\textsc{sng}}
\newcommand{\map}{\textsc{map}}
\newcommand{\flatten}{\textsc{flatten}}
\newcommand{\flatmap}{\textsc{flatmap}}
\newcommand{\pwith}[1]{\textsc{pairwith}_{#1}}
\newcommand{\close}{\textsc{close}}
\newcommand{\oref}{\textsc{ref}}
\newcommand{\dref}{\textsc{dref}}
\newcommand{\etag}{\textsc{tag}}
\newcommand{\att}{\textsc{att}}
\newcommand{\agg}{\textsc{agg}}
\newcommand{\child}{\textsc{child}}
\newcommand{\get}{\textsc{get}}
\newcommand{\nest}{\textsc{nest}}
\newcommand{\app}{\textsc{append}}
\newcommand{\push}{\ensuremath{\Phi}}
\newcommand{\val}{\textsc{V}}

\newcommand{\mrmap}{\ensuremath{\mathsf{map}}}
\newcommand{\mrreduce}{\ensuremath{\mathsf{reduce}}}
\newcommand{\s}{\ensuremath{\mathbf{s}}}
\newcommand{\e}{\ensuremath{\mathbf{e}}}
\newcommand{\f}{\ensuremath{\mathbf{f}}}
\newcommand{\Part}{\ensuremath{P}}
\newcommand{\Pclass}{\ensuremath{\mathbf{P}}}
\newcommand{\loc}{\ensuremath{\ell}}
\newcommand{\VR}{\ensuremath{VR}}
\newcommand{\oid}{\ensuremath{\mathsf{oid}}}
\newcommand{\rdom}{\ensuremath{\mathcal{L}}}

\begin{abstract}
  In many scientific domains, researchers are turning to large-scale
  behavioral simulations to better understand important real-world
  phenomena. While there has been a great deal of work on simulation tools
  from the high-performance computing community, behavioral simulations remain
  challenging to program and automatically scale in parallel environments. In
  this paper we present BRACE (Big Red Agent-based Computation Engine), which
  extends the MapReduce framework to process these simulations efficiently
  across a cluster.  We can leverage spatial locality to treat behavioral
  simulations as iterated spatial joins and greatly reduce the communication
  between nodes. In our experiments we achieve nearly linear scale-up on
  several realistic simulations.

  Though processing behavioral simulations in parallel as iterated spatial joins can be
  very efficient, it can be much simpler for the domain scientists to program
  the behavior of a single agent. Furthermore, many simulations include
  a considerable amount of complex computation and message passing between agents, 
  which makes it important to optimize
  the performance of a single node and the communication across nodes.
  To address both of these challenges, BRACE includes a high-level language
  called BRASIL (the Big Red Agent SImulation Language).  BRASIL has object
  oriented features for programming simulations, but can be compiled to a
  data-flow representation for automatic parallelization and optimization. We
  show that by using various optimization techniques, we can achieve both
  scalability and single-node performance similar to that of a hand-coded
  simulation.
\end{abstract}


\section{Introduction}
\label{sec:intro}

Behavioral simulations, also called agent-based simulations, are
instrumental in tackling the ecological and infrastructure
challenges of our society. These simulations allow scientists to
understand large complex systems such as transportation networks,
insect swarms, or fish schools by modeling the behavior of millions
of individual agents inside the
system~\cite{Buhl06locusts,CTB04,CKFL05}.

For example, transportation simulations are being used to address
traffic congestion by evaluating proposed traffic management
systems before implementing them~\cite{CTB04}. This is a
tremendously important problem as traffic congestion cost \$87.2
billion and required 2.8 billion gallons of extra fuel and 4.2
billion hours of extra time in the U.S. in 2007 alone~\cite{SL09}.
Scientists also use behavioral simulations to model collective
animal motion, such as that of locust swarms or fish
schools~\cite{Buhl06locusts,CKFL05}. Understanding these phenomena
is crucial, as they directly affect human food security~\cite{Gru06}.

Despite their huge importance, it remains difficult to develop
large-scale behavioral
simulations. Current systems either offer
high-level programming abstractions, but are not scalable~\cite{GVH03,LCP+05,MBLA96}, or achieve
scalability by hand-coding particular simulation models using low-level parallel frameworks, such as MPI~\cite{Wen08}.

This paper proposes to close this gap by bringing database-style
programmability and scalability to agent-based simulations. Our core
insight is that behavioral simulations may be regarded as
computations driven by large \emph{iterated spatial joins}. We
introduce a new simulation engine, called BRACE (Big Red
Agent-based Computation Engine), that extends the popular MapReduce
dataflow programming model to these iterated computations. BRACE
embodies a high-level programming language called BRASIL, which is
compiled into an optimized shared-nothing, in-memory MapReduce
runtime. The design of BRACE is motivated by the requirements of
behavioral simulations, explained below.

\subsection{Requirements for Simulation Platforms}
\label{sec:behavioral:requirements}

We have identified several key features that are necessary for a generic
behavioral simulation platform.

\minisec{(1)~Support for Complex Agent Interaction} Behavioral
simulations include frequent local interactions between agents. In
particular, agents may affect the behavior decisions of other
agents, and multiple agents may issue concurrent writes to the same
agent. A simulation framework should support a high degree of agent
interaction without excessive synchronization or rollbacks. This
precludes discrete event simulation engines or other approaches
based on task parallelism and asynchronous message exchange.

\minisec{(2)~Automatic Scalability} Scientists need to scale their simulations to
millions or billions of agents to accurately model phenomena such as city-wide
traffic or swarms of insects~\cite{Buhl06locusts,CKFL05,EFSC09}. These scales
make it essential to use data parallelism to distribute agents across many
nodes. This is complicated by the interaction between agents, which may
require communication between several nodes. Rather than requiring scientists
to write complex and error-prone parallel code, the platform should automatically
distribute agents to achieve scalability.

\minisec{(3)~High Performance} Behavioral simulations are often
extremely complex, involving sophisticated numerical computations
and elaborate decision procedures. Much existing work on behavioral
simulations is from the high-performance computing community, and
they frequently resort to hand-coding specific simulations in a
low-level language to achieve acceptable
performance~\cite{EFSC09,NR01}. A general purpose framework must be
competitive with these hand-coded applications in order to gain
acceptance.

\minisec{(4)~Commodity Hardware} Historically, many scientists have used
large shared-memory supercomputer systems for their
simulations. Such machines are tremendously expensive, and cannot scale
beyond their original capacity. We believe that the next generation of
simulation platforms will target shared-nothing systems and will be deployed
on local clusters or in the cloud on services like Amazon's EC2~\cite{AmazonEC2}.

\minisec{(5)~Simple Programming Model} Domain scientists have shown
their willingness to try simulation platforms that provide simple,
high-level programming abstractions, even at some cost in performance
and scalability
\cite{GVH03,LCP+05,MBLA96}.
Nevertheless, a behavioral simulation framework should
provide an expressive and high-level programming model without
sacrificing performance.

\subsection{Contributions}
\label{sec:contributions}

We begin our presentation in Section~\ref{sec:behavioral} by describing
important properties of behavioral simulations that we leverage in our
platform. We then move on to the main contributions of this paper:

\vspace*{-0.5ex}
\begin{itemize}

\item We show how MapReduce can be used to scale behavioral simulations across
  clusters. We abstract these simulations as iterated spatial joins and introduce
  a new main memory MapReduce runtime that incorporates optimizations motivated by the spatial properties of simulations (Section~\ref{sec:mapreduce}).
 \vspace*{-1ex}
\item We present a new scripting language for simulations that compiles into
  our MapReduce framework and allows for algebraic optimizations in mappers
  and reducers. This language hides all the complexities of modeling
  computations in MapReduce and parallel programming from domain scientists
  (Section~\ref{sec:programming}).
\vspace*{-1ex}
\item We perform an experimental evaluation with two real-world behavioral
  simulations that shows our system has nearly linear scale-up and single-node
  performance that is comparable to a hand-coded simulation.
  (Section~\ref{sec:exps}).
\end{itemize}
\vspace*{-0.5ex}

We review related work in Section~\ref{sec:related} and conclude in
Section~\ref{sec:conclusions}.

\section{Background}
\label{sec:background}

In this section, we introduce some important properties of behavioral
simulations and review the MapReduce programing model.  In the next section we
exploit these properties to abstract the computations in the behavioral
simulations and use MapReduce to efficiently process them.

\subsection{Properties of Behavioral Simulations}
\label{sec:behavioral}

Behavioral simulations model large numbers of individual agents that interact
in a complex environment. Unlike scientific simulations that can be modeled as
systems of equations, agents in a behavioral simulation can execute complex
programs that include non-trivial control flow. Nevertheless, most behavioral
simulations have a similar structure, which we introduce below. We will use a
traffic simulation~\cite{YKB99} and a
fish school simulation~\cite{CKFL05} as running examples. Details on these
simulation models can be found in Appendix~\ref{sec:exps:model:details}.

\minisec{The State-Effect Pattern} Most behavioral simulations use a
time-stepped model in which time is discretized into ``ticks'' that represent
the smallest time period of interest. Events that occur during the same tick
are treated as simultaneous and can be reordered or parallelized. This means
that an agent's decisions cannot be based on previous actions made during the
same tick. An agent can only read the state of the world as of the previous
tick.
For example, in the traffic simulation, each car inspects the
positions and velocities of other cars as of the beginning of the
tick in order to make lane changing decisions.

In previous work on scaling computer games, we proposed a model for this kind
of time-stepped behavior called the
\emph{state-effect pattern}~\cite{White07:Epic,White08:Sandbox}.
The basic idea is to separate
read and write operations in order to limit the synchronization necessary
between agents. In the state-effect pattern, the attributes of an agent are
separated into \emph{states} and \emph{effects}, where states are
public attributes that are updated only at tick boundaries, and effects are
used for intermediate computations as agents interact. Since state attributes
remain fixed during a tick, they only need to be synchronized at the end of
each tick. Furthermore, each effect attribute has an associated decomposable
and order-independent \emph{combinator function} for combining multiple
assignments during a tick.
This allows us to compute effects in parallel and combine the results without
worrying about concurrent writes. For example, the fish school
simulation uses vector addition so that each fish may combine the orientation of
nearby fish into an effect attribute.
Since vector addition is commutative, we can process these assignments in any order.

\vspace*{-1ex}
\parpic[r]{\scalebox{0.22}{\includegraphics[trim=320 115 315 90,clip=true]%
            {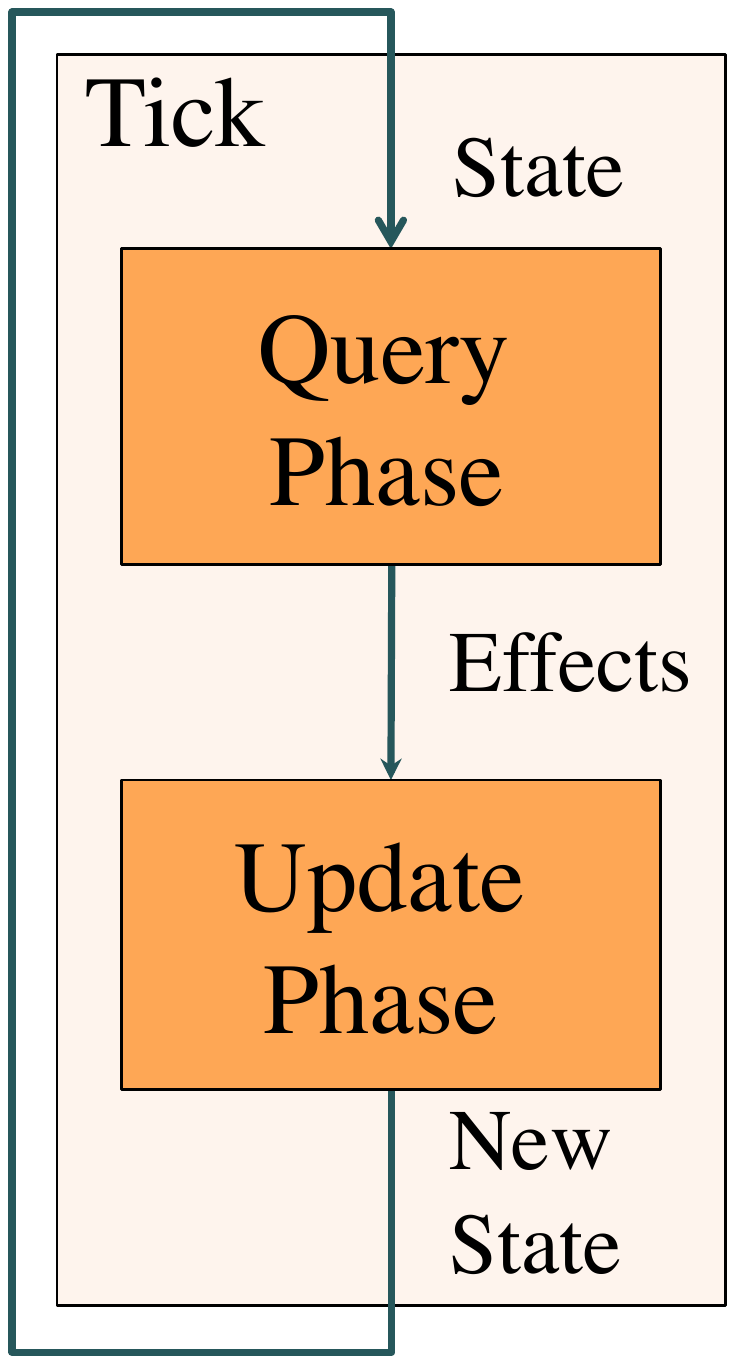}} }
In the state-effect pattern, each tick is divided into two phases: the
\emph{query phase} and the \emph{update phase}, as shown in the figure
on the right. In the query phase, each agent queries the state of the world
and assigns effect values, which are combined using the appropriate combinator
function. To ensure the property that the actions during a tick are conceptually
simultaneous, state variables are read-only during the query phase and effect
variables are write-only.

In the update phase, each agent can read its state attributes and the effect
attributes computed from the query phase; it uses these values to compute
the new state attributes for the next tick. In the fish school simulation,
the orientation effects computed during the query phase are read during the
update phase to compute a fish's new velocity vector, represented as a state
attribute.
In order to ensure that updates do not conflict, each agent can only read and
write its own attributes during the update phase. Hence, the only way that agents
can communicate is through effect assignments in the query phase.
We classify effect assignments into \emph{local} and \emph{non-local} assignments.
In a local assignment, an agent updates one of its own effect attributes;
in a non-local assignment, an agent writes to an effect
attribute of a different agent.

\minisec{The Neighborhood Property} The state-effect pattern can be used to
limit the synchronization necessary during a tick, but it is still possible
that every agent needs to query every other agent in the simulated world to compute its effects. We
observe that this rarely occurs in practice. In particular, we observe that
most behavioral simulations are eminently spatial, and simulated agents can
only interact with other agents that are close according to a distance metric~\cite{EA96}. For
example, a fish can only observe other fish within a limited distance $\rho$.

\subsection{MapReduce}
\label{sec:back:mapreduce}

Since its introduction in 2004, MapReduce has become one of the most
successful models for processing long running computations in
distributed shared-nothing environments\cite{Dean04:MapReduce}.
While it was originally designed for very large batch computations,
MapReduce is ideal for behavioral simulations because it provides
automatic scalability, which is one of the key requirements for
next-generation platforms. By varying the degree of data
partitioning and the corresponding number of map and reduce tasks,
the same MapReduce program can be run on one machine or one
thousand.

The MapReduce programming model is based on two functional
programming primitives that operate on key-value pairs. The
\textsf{map} function takes a key-value pair and produces a set of
intermediate key-value pairs, $\textsf{map}: (k_1,v_1) \to
[(k_2,v_2)]$, while the reduce function collects all of the
intermediate pairs with the same key and produces a value,
$\textsf{reduce}: (k_2,[v_2]) \to [v_3]$. Since simulations consist
of many ticks, we will use an iterative MapReduce model in which the
output of the reduce step is fed into the next map step. Formally,
this means that we change the output of the reduce step to be
$[(k3,v3)]$.

\section{MapReduce for Simulations}
\label{sec:mapreduce}

In this section, we abstract behavioral simulations as computations
driven by iterated spatial joins and show how they can be expressed
in the MapReduce framework
(Section~\ref{sec:iterative:spatial:joins}). We then propose a
system called BRACE to process these joins efficiently
(Section~\ref{sec:mrruntime}).

\subsection{Simulations as Iterated Spatial Joins}
\label{sec:iterative:spatial:joins}

\begin{table}
\small
\begin{tabular}{|c|c|c|c|c|c|} \hline
  \textbf{effects} & $\mathbf{map_1^t}$ & $\mathbf{reduce_1^t}$ & $\mathbf{map_2^t}$ & $\mathbf{reduce_2^t}$ & $\mathbf{map_1^{t+1}}$ \\ \hline
  \multirow{2}{*}{\textbf{local}} & update$^{t-1}$  & \multirow{2}{*}{query$^{t}$}  & \multirow{2}{*}{---} & \multirow{2}{*}{---} & update$^{t}$ \\
  & distribute$^t$ &  &  & & distribute$^{t+1}$ \\ \hline
  \textbf{non-} & update$^{t-1}$ & local & \multirow{2}{*}{---} & global & update$^{t}$ \\
  \textbf{local} & distribute$^t$ & effect$^t$ &  &  effect$^t$ & distribute$^{t+1}$ \\ \hline
\end{tabular}
\vspace*{-2ex}
\caption{The state-effect pattern in MapReduce.}
\label{tab:mapreduce}
\vspace*{-3.5ex}
\end{table}

In Section~\ref{sec:behavioral}, we observed that behavioral
simulations have two important properties: the state-effect pattern
and the neighborhood property. The state-effect pattern essentially
characterizes behavioral simulations as iterated computations with
two phases: a query phase in which agents inspect their environment
to compute effects, and an update phase in which agents update their
own state.

The neighborhood property introduces two important restrictions on each of
these phases, visibility and reachability. We say that the
\emph{visible region} of an agent $a$ is the region of space containing agents that $a$ can
read from or assign effects to. Agent $a$ needs access to all the agents in
its visible region to compute its query phase. Thus a simulation in which
agents have small visible regions requires less communication
than one with very large or unbounded visible regions.  Similarly, we can
define an agent's \emph{reachable region} as the region that the agent can
move to after the update phase. This is essentially a measure of how much the
spatial distribution of agents can change between ticks. When agents have
small reachable regions, a spatial partitioning of the agents is
likely to remain balanced for several ticks. Frequently an agent's reachable
region will be a subset of its visible region (an agent can't move farther
than it can see), but this is not required.

We observe that since agents only query other agents within their visible
regions, processing a tick is similar to a \emph{spatial self-join} from the
database literature~\cite{LNE02}. We join each agent with the set of agents in
its visible region and perform the query phase using only these agents. During
the update phase, agents move to new positions within their reachable regions
and we perform a new iteration of the join during the next tick. We will use
these observations to parallelize behavioral simulations efficiently in the
MapReduce framework.

\subsection{Iterated Spatial Joins in MapReduce}
\label{sec:iterative:spatial:joins:mapreduce}

In this section, we show how to model spatial joins in MapReduce. A formal
version of this model is included in
Appendix~\ref{sec:joins:mapreduce}. MapReduce has often been criticized for
being inefficient at processing joins~\cite{YangDHRP07:MapReduceMerge} and
also inadequate for iterative computations without
modification~\cite{Twister:eScience08}. However, the spatial properties of
simulations will allow us to process them effectively without excessive
communication. Our basic strategy is to use a technique presented by Zhang et
al.~to compute a spatial join in MapReduce~\cite{Zhang09:CLUSTER}. Each map
task is responsible for spatially partitioning agents into a number of
disjoint regions, and the reduce tasks join the agents using their visible
regions.

The map tasks use a spatial partitioning function to assign each agent to a
disjoint region of space. This function might divide the space into a regular
grid or might perform some more sophisticated spatial decomposition. Each
reducer will process one such partition. The set of agents assigned to a
particular partition is called that partition's \emph{owned set}. Note that we
cannot process each partition completely independently because each agent
needs access to its visible region, and this region may intersect several
partitions. To address this, we can define the \emph{visible region} of a
partition as the region of space visible to some point in the partition. The
map task will then \emph{replicate} each agent $a$ to every partition that
contains $a$ in its visible region.

Table~\ref{tab:mapreduce} shows how the phases of the state-effect pattern are
assigned to map and reduce tasks. For simulations with only local effect
assignments, a tick $t$ begins when the first map task, $\mrmap_1^t$, assigns
each agent to a partition (distribute$^t$). Each reducer is assigned a
partition and receives every agent that falls within its owned set as well as
replicas of every agent that falls within its visible region. These are
exactly the agents necessary to process the query phase of the owned set
(query$^t$). The reducer, $\mrreduce_1^{t}$, outputs a copy of each agent it
owns after executing the query phase and updating the agent's effects. The
tick ends when the next map task, $\mrmap_1^{t+1}$, executes the update phase
(update$^{t}$).

While this two-step approach works well for simulations that have only local
effects, it does not handle non-local effect assignments. Recall that a
non-local effect is an effect assignment by some agent $a$ to some other agent
$b$ within $a$'s visible region. For example, if we were to extend the fish
simulation to include predators, then we would model a shark attack as a
non-local effect assignment from the shark to the fish. Non-local effects
require communication during the query phase. We implement this communication
using two MapReduce passes, as illustrated in Table~\ref{tab:mapreduce}. The
first map task, $\mrmap_1^{t}$, is the same as before. The first reduce task,
$\mrreduce_1^{t}$, performs only effect assignments to its local replicas
(local effect$^t$). These partially aggregated effect values are then
distributed to the partitions that own them, where they are combined by the
second reduce task, $\mrreduce_2^{t}$. This computes the final value for each
aggregate (global effect$^t$). As before, the update phase is processed in the
next map task, $\mrmap_1^{t+1}$. Note that the second map task,
$\mrmap_2^{t}$, is only necessary for distribution, but does not perform any
computation and can be eliminated in an implementation. We call this model
map-reduce-reduce.

Our map-reduce-reduce model relies heavily on the neighborhood property. The
number of replicas that each map task must create depends on the size of the
agent's visible regions, and the frequency with the agents change partitions
depends on the size of their reachable regions.

\subsection{The BRACE MapReduce Runtime}
\label{sec:mrruntime}

In this section we describe a MapReduce implementation that takes
advantage of the state-effect pattern and the neighborhood property.
We introduce BRACE, the Big Red Agent Computation Engine, our
platform for scalable behavioral simulations. BRACE includes a
MapReduce runtime especially optimized for the iterated spatial
joins discussed in Section~\ref{sec:iterative:spatial:joins}.
We have developed a new system rather than using an
existing MapReduce implementation such as Hadoop~\cite{Hadoop} because behavioral
simulations have considerably different characteristics than traditional
MapReduce applications such as search log analysis.
The goal of BRACE is to process a
very large number of ticks efficiently, and to avoid I/O or communication
overhead while providing features such as fault tolerance.
We describe the main features of
our runtime below.

\minisec{Shared-Nothing, Main-Memory Architecture} In behavioral simulations,
we expect data volumes to be modest, so BRACE executes map and reduce tasks
entirely in main memory. For example, a simulation with one million agents
whose state and effect fields occupy 1~KB on average requires roughly 1~GB of
main memory. Even larger simulations with orders of magnitude more agents will
still fit in the aggregate main memory of a cluster. Since
the query phase is computationally expensive, partition sizes are limited by
CPU cycles rather than main memory size.

\begin{figure}[!t]
\centering
\includegraphics[width=0.38\textwidth]{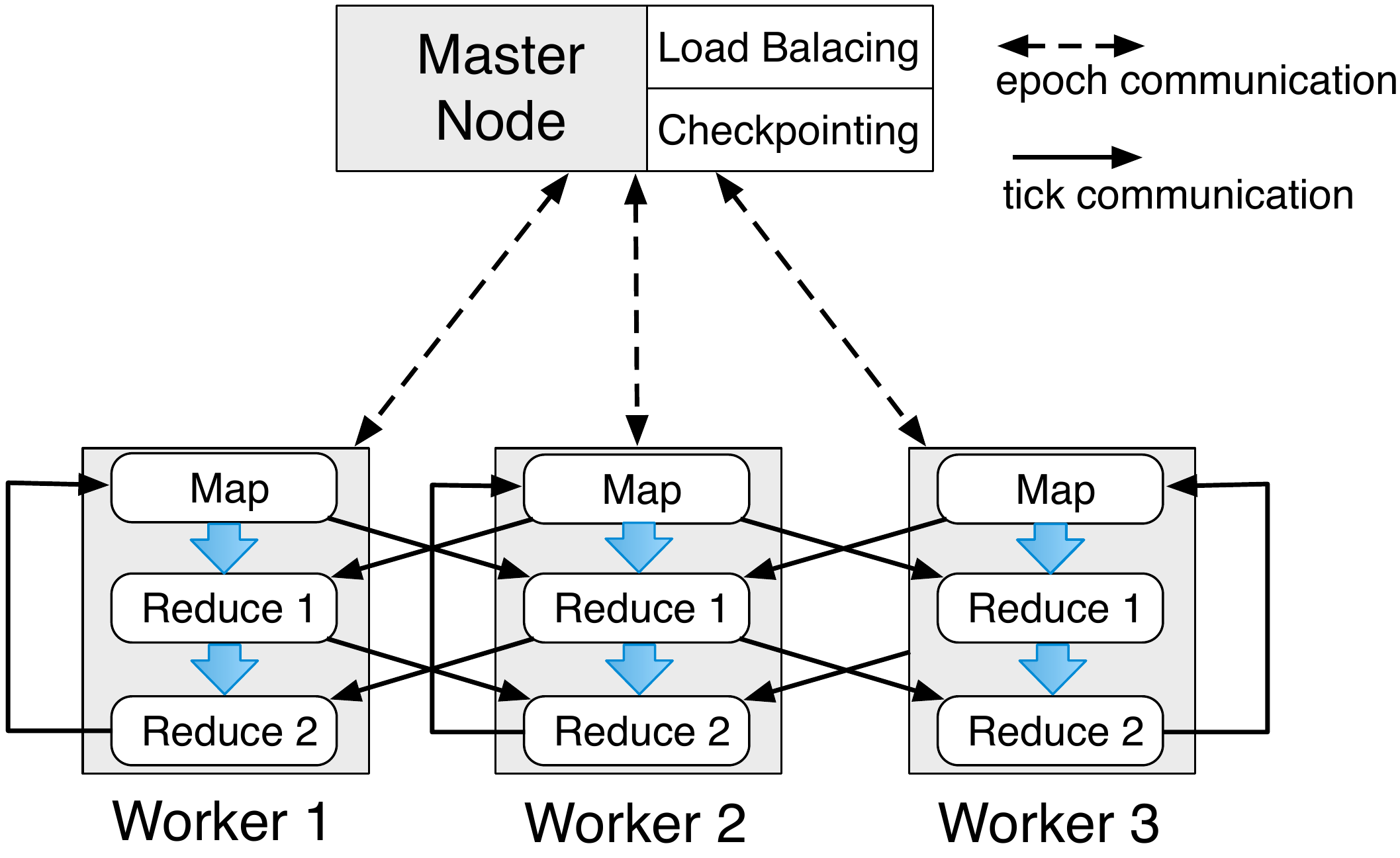}
\vspace*{-2ex} \caption{\label{fig:framework} BRACE Architecture
Overview} \vspace*{-4ex}
\end{figure}

Figure~\ref{fig:framework} shows the architecture of BRACE. As in typical
MapReduce implementations, a master node is responsible for cluster
coordination. Unlike in traditional MapReduce runtimes, BRACE's master node
only interacts with worker nodes every \emph{epoch}, which corresponds to a
fixed number of ticks. The intuition is that iterations will be quickly
processed in main memory by the workers, so we wish to amortize the overheads
related to fault tolerance and load balancing. In addition, we carefully
allocate tasks of map-reduce-reduce iterations to workers, so that we diminish
communication overheads within and across iterations.

\minisec{Fault Tolerance}
Traditional MapReduce runtimes provide fault tolerance
by storing output to a replicated file system and
automatically restarting failed tasks.
Since we expect ticks to be quite short and they are processed in main memory,
it would be prohibitively expensive to write output to stable storage between
every tick. Furthermore, since individual ticks are short, the benefit from
restarting a task is likely to be small.

We employ epoch synchronization with the master to trigger \emph{coordinated checkpoints}~\cite{EAWJ99}
of the main memory of the workers. As the master determines a pre-defined tick boundary for checkpointing,
the workers can write their checkpoints independently without global synchronization.~As we expect iterations to be short,
failures are handled by
re-execution of all iterations since the last checkpoint, a common technique in scientific simulations.
In fact, we can leverage previous literature to tune the checkpointing interval to minimize
the total expected runtime of the whole computation~\cite{Dal06}.

\minisec{Partitioning and Load Balancing} As we have observed in
Section~\ref{sec:iterative:spatial:joins:mapreduce}, bounded reachability
implies that a given spatial partitioning will remain effective for a number
of map-reduce-reduce iterations.  Our runtime uses that observation to keep
data partitioning stable over time and re-evaluates it at epoch boundaries.

At the beginning of the simulation, the master computes a
partitioning function based on the visible regions of the agents and
then broadcasts this partitioning to the worker nodes. Each worker
becomes responsible for one region of the partitioning. While agents
change partitions slowly, over time the overall spatial distribution
may change quite dramatically. For example, the distribution of
traffic on a road network is likely to be very different at morning
rush hour than at evening rush hour. This would cause certain nodes
to become overloaded if we used the same partitioning in both cases.
To address this, the master periodically requests statistics from
the workers about the number of agents in the owned region and the
communication and processing costs. The master then decides on
repartitioning by balancing the cost of redistribution with its
expected benefit. If the master decides to modify the partitioning,
it broadcasts the new partitioning to all workers. The workers then
switch to the new partitioning at a specified epoch boundary.

\minisec{Collocation of Tasks}
Since simulations run for many iterations, it
is important to avoid unnecessary communication between map and reduce
tasks. We accomplish this by collocating the map and reduce tasks for a tick
on the same node so that agents that do not switch partitions can be sent
between tasks via fast memory rather than the network. Since agents have limited
reachable regions, the owned set of each partition is likely to remain
relatively stable across ticks, and so will remain on the same node.  Agents
still need to be replicated, but their original copies do not have to be
redistributed. This idea was previously explored by the Phoenix project for
SMP systems~\cite{YooRK09:phoenix2} and the Map-Reduce-Merge project for
individual joins~\cite{YangDHRP07:MapReduceMerge}, but it is particularly
important for long-running behavioral simulations.

Figure~\ref{fig:framework} shows how collocation works when we allow non-local
effect assignments. Solid arrows indicate the flow of agents during a tick.
Each node processes a map task and two reduce tasks as described in
Section~\ref{sec:iterative:spatial:joins}. The map task replicates agents as appropriate and
sends all of the agents that remain in the same partition to the reduce task
on the same node.  The first reducer computes local effects and sends any
updated replicas to the second reduce phase at other nodes.  The final reducer
computes the final effects and sends them to the map task on the same
node. Because of the neighborhood property, many agents will be
processed on the same node during the next tick.

\section{Programming Agent Behavior}
\label{sec:programming}

In this section, we show how to offer a simple programming model for a domain scientist,
targeting the last requirement of Section~\ref{sec:behavioral:requirements}. MapReduce
is set-based; a program describes how to process all of the elements in a collection.
Simulation developers prefer to describe the behavior of their agents individually, and use
message-passing techniques to communicate between agents. This type of
programming is closer to the scientific models that describe agent behavior.

We introduce a new programming language -- BRASIL, the Big Red Agent
SImulation Language. BRASIL embodies agent centric programming with
explicit support for the state-effect pattern, and performs further
algebraic optimizations. It bridges the mental model of simulation
developers and our MapReduce processing techniques for behavioral
simulations. We provide an overview of the main features of BRASIL
(Section~\ref{sec:programming:brasil}) and describe algebraic
optimization techniques that can be applied to our scripts
(Section~\ref{sec:programming:opt}). Formal semantics for our
language as well as the proofs of theorems in this section are
provided in Appendix~\ref{sec:monad:formal}.

\subsection{Overview of BRASIL}
\label{sec:programming:brasil}

BRASIL is an object-oriented language in which each object corresponds
to an agent in the simulation. Agents in BRASIL are defined in a
class file that looks superficially similar to Java.  The programmer
can specify fields, methods, and constructors, and these can each
either be public or private. Unlike in Java, however, each field in
a BRASIL class must be tagged as either \emph{state} or
\emph{effect}. The BRASIL compiler then enforces the read-write
restrictions of the state-effect pattern over those fields as
described in Section~\ref{sec:behavioral}. Figure~\ref{Fi:Class}
illustrates an example of a simple two-dimensional fish simulation; 
in this simulation, the fish swim about randomly, but avoid each other 
through the use of imaginary repulsion ``forces''.

\begin{figure}[t]
\small
\begin{verbatim}
class Fish {
  // The fish location
  public state float x : (x+vx); #range[-1,1];
  public state float y : (y+vy); #range[-1,1];

  // The latest fish velocity
  public state float vx : vx + rand() + avoidx / count * vx;
  public state float vy : vy + rand() + avoidy / count * vy;

  // Used to update our velocity
  private effect float avoidx : sum;
  private effect float avoidy : sum;
  private effect int count : sum;

  /** The query-phase for this fish. */
  public void run() {
    // Use "forces" to repel fish too close
    foreach(Fish p : Extent<Fish>) {
       p.avoidx <- 1 / abs(x - p.x);
       p.avoidy <- 1 / abs(y - p.y);
       p.count  <- 1;
    }}}
\end{verbatim}
\vspace{-5ex}
\caption{Class for Simple Fish Behavior}
\label{Fi:Class}
\vspace{-5ex}
\end{figure}

Recall that the state-effect pattern divides the computation into a query and
an update phase. In BRASIL, the query phase for an agent class is expressed by
its \texttt{run()} method. State fields remain read-only and all effect assignments
are aggregated using the aggregate function specified at the effect field
declaration. Effect fields are similar to aggregator variables in
Sawzall~\cite{Pike05:Sawzall}; indeed, we use the Sawzall operator \texttt{<-}
for writing to effect fields.
In our fish simulation, for example, each fish repels nearby fish
via a ``force'' inversely proportional to the distance between them.
The update phase is specified as a collection of 
\emph{update rules} attached to each state field.  These rules can only 
read values of other fields in this agent. In our example, fish velocity 
vectors get updated based on the avoidance factors and then perturbed 
by a random amount.

There are some important restrictions in BRASIL's programming constructs. First,
BRASIL only supports iteration over a set or list via a \texttt{foreach}-loop.
This eliminates arbitrary looping, which is not available in algebraic database
languages.
Second, there is an interplay between \texttt{foreach}-loops and effects: effect variables can only be read outside of a
\texttt{foreach}-loop, and all assignments within a \texttt{foreach}-loop are
aggregated. This powerful restriction allows us to treat the entire program,
and not just the communication across map and reduce operations, as a
data-flow query plan.

BRASIL also has a special programming construct to enforce the neighborhood
property outlined in Section~\ref{sec:behavioral}. Every state field that
encodes spatial location may be tagged with a visibility and reachability
constraint. While it is possible to generalize this concept to arbitrary
constraints, in our current implementation the constraints are
(hyper)rectangles. For example, the constraint attached to the
\texttt{x} field in Figure~\ref{Fi:Class} means that $[-1,1]$ is the 
interval that this fish can inspect or move with respect to the \texttt{x}
coordinate. In our fish example, the \texttt{foreach}-loop will therefore
only be able to affect fish within this range. In addition, the update
rule is guaranteed to crop any changes to the \texttt{x} coordinate to at
most one unit.

Note that visibility has an interplay with agent references: it is possible
that a reference to another agent is fine initially, but violates the
visibility constraint as that other agent moves relative to the one holding
the reference.  For that reason, BRASIL employs
\emph{weak reference semantics} for agent references,
similar to weak references in Java.  If another agent moves outside of the
visible region, then all references to it will resolve to \nil.

Note that this gives a different semantics for visibility than the one
present in Section~\ref{sec:mapreduce}. BRASIL uses visibility to determine
how agent references are resolved, while the BRACE runtime uses visibility to
determine agent replication and communication.  The BRASIL semantics are
preferable for a developer, because they are easy to understand
and hide MapReduce details. Fortunately, as we prove formally in
Appendix~\ref{sec:monad:formal}, these are equivalent.

\vspace{-1ex}
\begin{theoremreference}{\ref{T:Visible}}
The BRASIL semantics for visibility and the BRACE implementation of visibility
are equivalent.
\end{theoremreference}
\vspace{-1ex}

While programming features in BRASIL may seem unusual, everything in the
language follows from the state-effect pattern and neighborhood property.  As
these are natural properties of behavioral simulations, programming these
simulations becomes relatively straightforward.  Indeed, a large part of our
traffic simulation in Section~\ref{sec:exps:setup} was implemented by a domain
scientist.

\subsection{Optimization}
\label{sec:programming:opt}

We compile BRASIL into a well-understood data-flow language. In our previous
work on computer games, we used the relational algebra to represent
our data flow~\cite{White07:Epic}. However, for distributed simulations,
we have found the monad algebra~\cite{BNTW1995,Koch2005,Paradaens88,TBW1992}
-- the theoretical foundation for XQuery~\cite{Koch2005} -- to be a much more
appropriate fit. In particular, the monad algebra has a $\map$ primitive for
descending into the components of its nested data model; this makes it a
much more natural companion to MapReduce than the relational algebra.

We present the formal translation to the monad algebra
in Appendix~\ref{sec:monad:overview}, together with several theorems
regarding its usage in optimization. Most of these optimizations are the same
as those that would be present in a relational algebra query plan: algebraic
rewrites and automatic indexing.  In fact, any monad algebra expression on
flat tables can be converted to an equivalent relational algebra expression
and {\em vice versa}~\cite{Paradaens88}; rewrites and indexing on the relational
form carry back into the monad algebra form. In particular, many of the
techniques used by Pathfinder~\cite{Pathfinder} to process XQuery with
relational query plans apply to the monad algebra.

\newenvironment{smallenv}{\small}{}
\minisec{Effect Inversion}
An important optimization that is unique to our framework involves the
assignment of non-local effects. If non-local effect assignments can
be eliminated, then we are able to process our MapReduce computations
with one MapReduce pass instead of two (Section~\ref{sec:mapreduce}).
Consider again the program of Figure~\ref{Fi:Class}. We may rewrite
its \texttt{foreach}-loop as

\vspace{-1.5ex}
\begin{smallenv}
\begin{verbatim}
  foreach(Fish p : Extent<Fish>) {
       avoidx <- 1 / abs(p.x - x);
       avoidy <- 1 / abs(p.y - y);
       count <- 1;
  }
\end{verbatim}
\end{smallenv}
\vspace{-1.5ex}
This rewritten expression does not change the results of the simulation, but
only assigns effects locally.

There are two main results regarding effect inversion.  While we prove them
formally in Appendix~\ref{sec:monad:inversion}, we can state them informally
here, and give some intuition regarding their usage.

\vspace{-2ex}
\begin{theoremreference}{\ref{T:Invert}}
Effect inversion is always possible if there are no visibility constraints.
\end{theoremreference}
\vspace{-1ex}

When there are no visibility constraints, each agent can read any other
agent. Hence, we can produce a script where an agent simulates the behavior
of other agents, checks for effects that are assigned to itself, and then
assign them locally.  Unfortunately, this new script is clearly much more
computationally expensive.  Hence, to be useful in practice, we need to
use other optimization techniques to simplify the new script.

Effect inversion is not always possible when the simulation has visibility
constraints. Intuitively, an agent may use non-local effects to act as
a communication proxy between two other agents that are not visible to
one another.  However, the state-effect pattern ensures that an agent
can only receive information from another agent if they have a third
(not necessarily distinct) agent that is visible to both.  This provides
us with another result.

\vspace{-2ex}
\begin{theoremreference}{\ref{T:Visible:Invert}}
If the visibility constraint on a script is a distance bound, there
is an equivalent script with a constraint at most twice that distance bound 
that has only local effect assignments.
\end{theoremreference}
\vspace{-1ex}

Increasing the visibility bound increases the number of replicas that have
to store at each node.  Hence this optimization eliminates the extra
communication round at the cost of more information to be sent during
the remaining communication round.

\section{Experiments}
\label{sec:exps}

In this section, we present experimental results using two distinct
real-world behavioral simulation models we have coded using BRACE.
We focus on the following:
(i)~We validate the effectiveness of the BRASIL optimizations introduced in Section~\ref{sec:programming:opt}.
In fact, these optimizations allow us to
approach the efficiency of hand-optimized simulation code (Section~\ref{sec:exps:single:node});
(ii)~We evaluate BRACE's MapReduce runtime implementation over a cluster of machines.
We measure simulation scale-up via spatial data partitioning as well as load balancing (Section~\ref{sec:exps:scalability}).

\subsection{Setup}
\label{sec:exps:setup}

\minisec{Implementation}
The prototype BRACE MapReduce runtime is implemented in C++ and uses MPI for inter-node communication.
Our BRASIL compiler is written in Java and directly generates C++ code
that can be compiled with the runtime.

Our prototype includes a generic KD-tree based spatial index capability~\cite{Bentley90:KDTree}.
We use a simple rectilinear grid partitioning scheme, which assigns each grid cell
to a separate slave node,
A one-dimensional load balancer periodically receives statistics from the slave nodes,
including computational load and number of owned agents;
from these it heuristically computes a new partition
trying to balance improved performance against estimated migration cost.
Checkpointing is not yet
integrated into BRACE's implementation.
We believe this is not a problem for the cluster sizes we evaluate, given the low likelihood of worker
failure during a computation.

We plan to integrate more sophisticated algorithms for all these components in
future work.
But our current prototype already demonstrates good performance and automatic
scaling of realistic behavioral simulations written in BRASIL.

\minisec{Simulation Workloads}
We have implemented realistic traffic and fish school
simulations in BRASIL. The traffic simulation includes the lane-changing and
acceleration models of the state-of-the-art, open-source MITSIM traffic
simulator~\cite{YKB99}. MITSIM is a single-node program, so we
compare its performance against our BRASIL reimplementation of
its model also running on a single node.
We simulate a linear segment of highway, and
scale-up the size of the problem by extending the length of the segment.

The fish simulation implements a recent model of
information flow in animal groups~\cite{CKFL05}.
In this model the ``ocean'' is unbounded, and the spatial distribution of fish
changes dramatically as ``informed individuals'' guide the movements
of others in the school.

Neither of these simulations uses non-local effect assignments;
therefore we need only a single reducer per node.
To evaluate our effect inversion optimization,
we modified the fish simulation to create a predator simulation
that uses non-local assignments. It is similar in spirit to artificial
society simulations~\cite{EA96}.
Appendix~\ref{sec:exps:model:details} describes these simulation models
in more detail.
We measure total simulation time in our single-node experiments and
tick throughput (agent ticks per second) when scaling up over multiple nodes.
In all measurements we eliminate start-up transients by
discarding initial ticks until a stable tick rate is achieved.

\minisec{Hardware Setup}
We ran all of our experiments in the Cornell Web Lab cluster~\cite{AAD+06}. The cluster
contains 60 nodes interconnected by a pair of 1 gigabit/sec Port Summit X450a Ethernet Switches.
Each node has two Quad Core Intel Xeon, 2.66GHz, processors with 4MB cache each and
16 GB of main memory.

\subsection{BRASIL Optimizations}
\label{sec:exps:single:node}

\begin{figure*}[tb]
    \begin{minipage}{0.33\linewidth}
        \centerline{\includegraphics[width=0.85\linewidth]{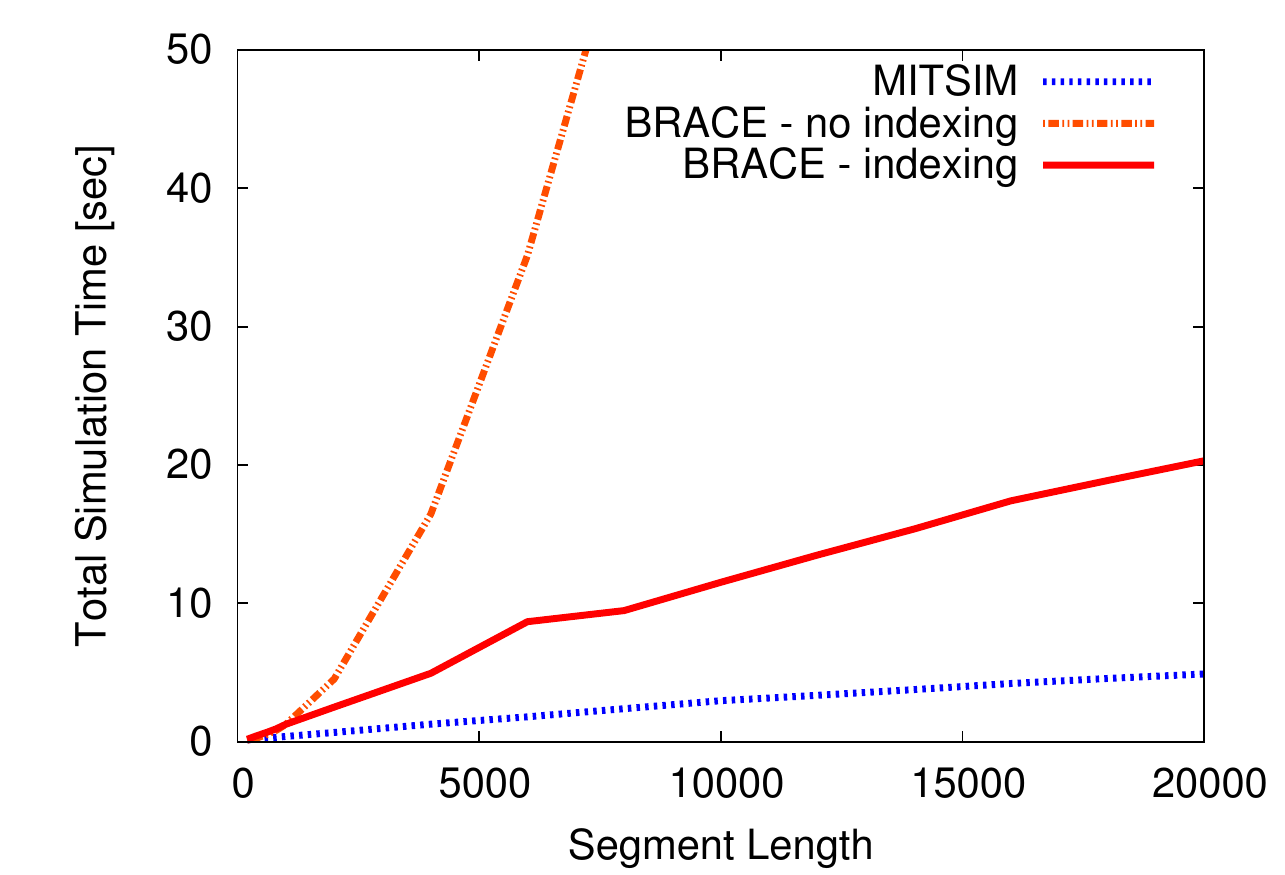}}
        \vspace{-1.5ex}
        \caption{Traffic: Indexing vs. Seg.~Length} \label{fig:traffic:segment:length}
    \end{minipage} \hfill
    \begin{minipage}{0.31\linewidth}
        \centerline{\includegraphics[width=0.9\linewidth]{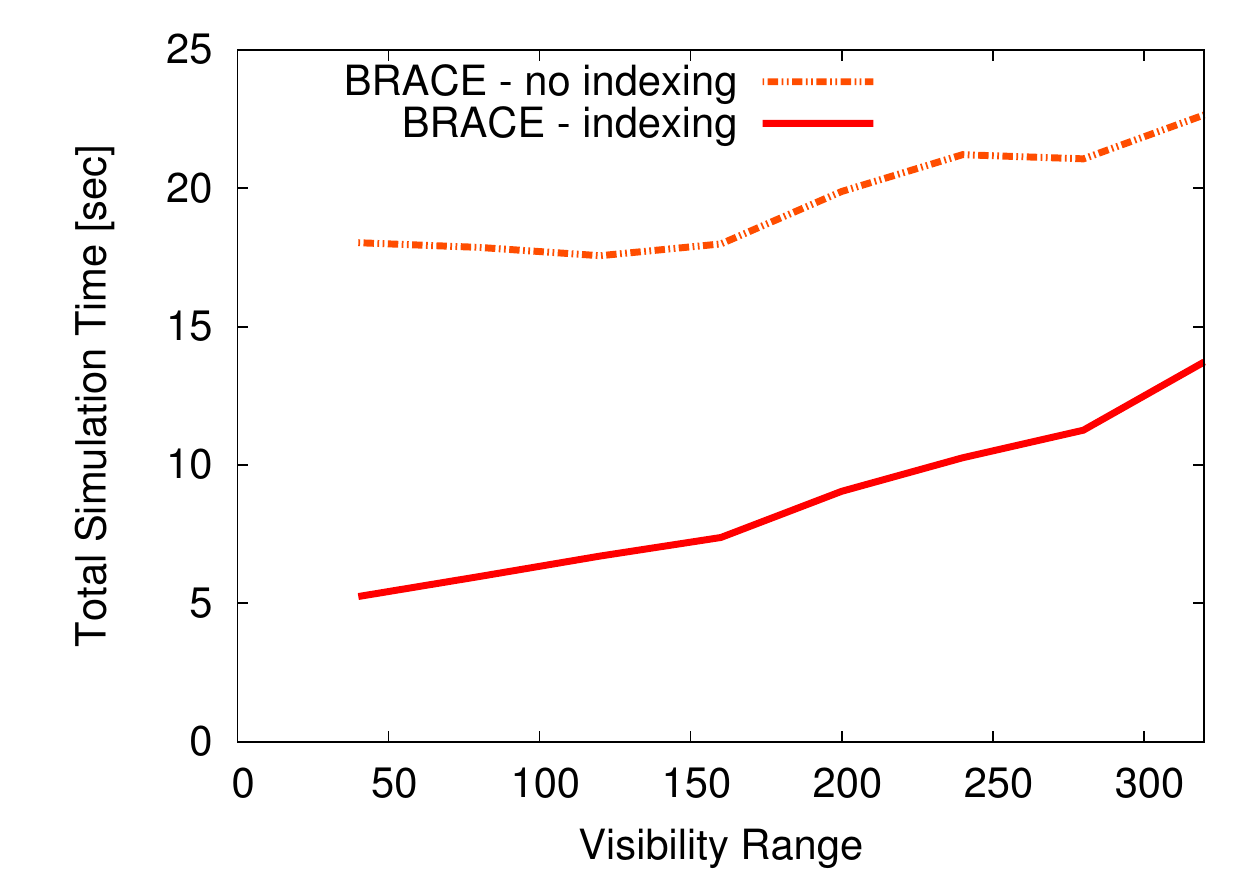}}
        \vspace{-1.5ex}
        \caption{Fish: Indexing vs. Visibility} \label{fig:fish:visibility}
    \end{minipage} \hfill
    \begin{minipage}{0.34\linewidth}
        \centerline{\includegraphics[width=0.9\linewidth]{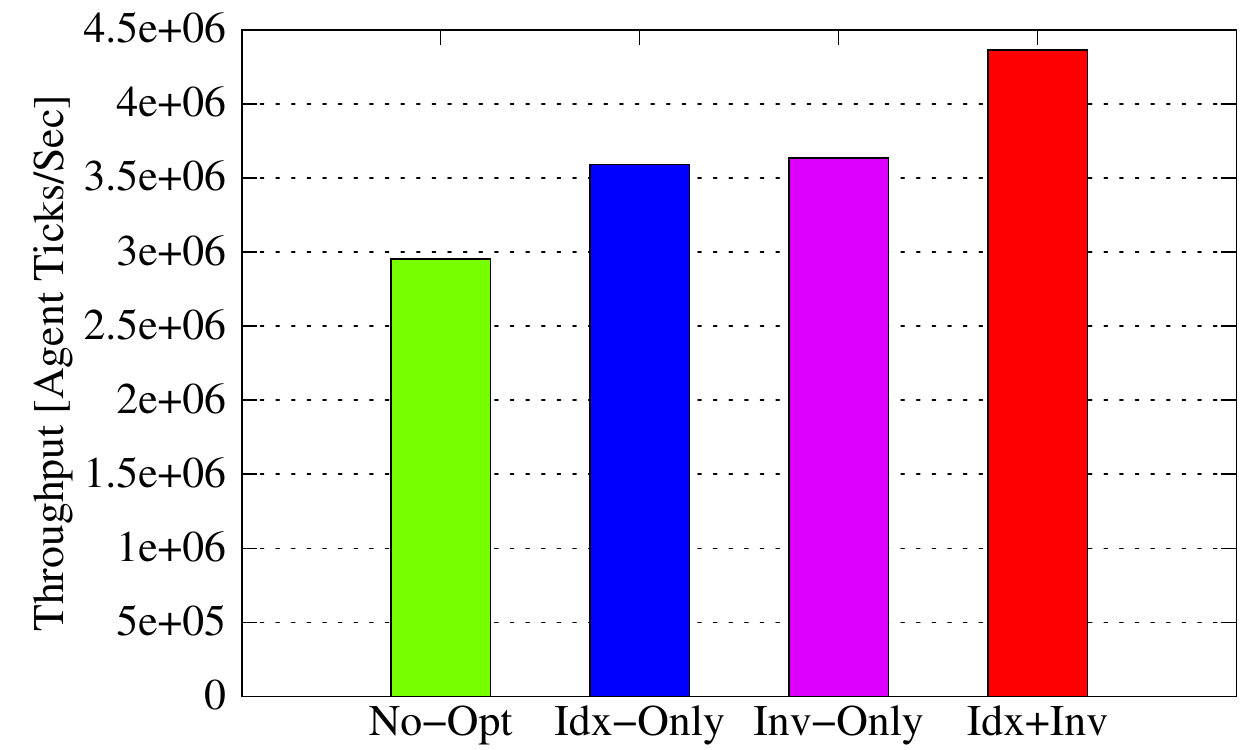}}
        \vspace{-0.5ex}
        \caption{Predator: Effect Inversion} \label{fig:predator:inversion}
    \end{minipage} \hfill
    \vspace{-1ex}
\end{figure*}
\begin{figure*}[tb]
    \begin{minipage}{0.327\linewidth}
        \centerline{\includegraphics[width=0.9\linewidth]{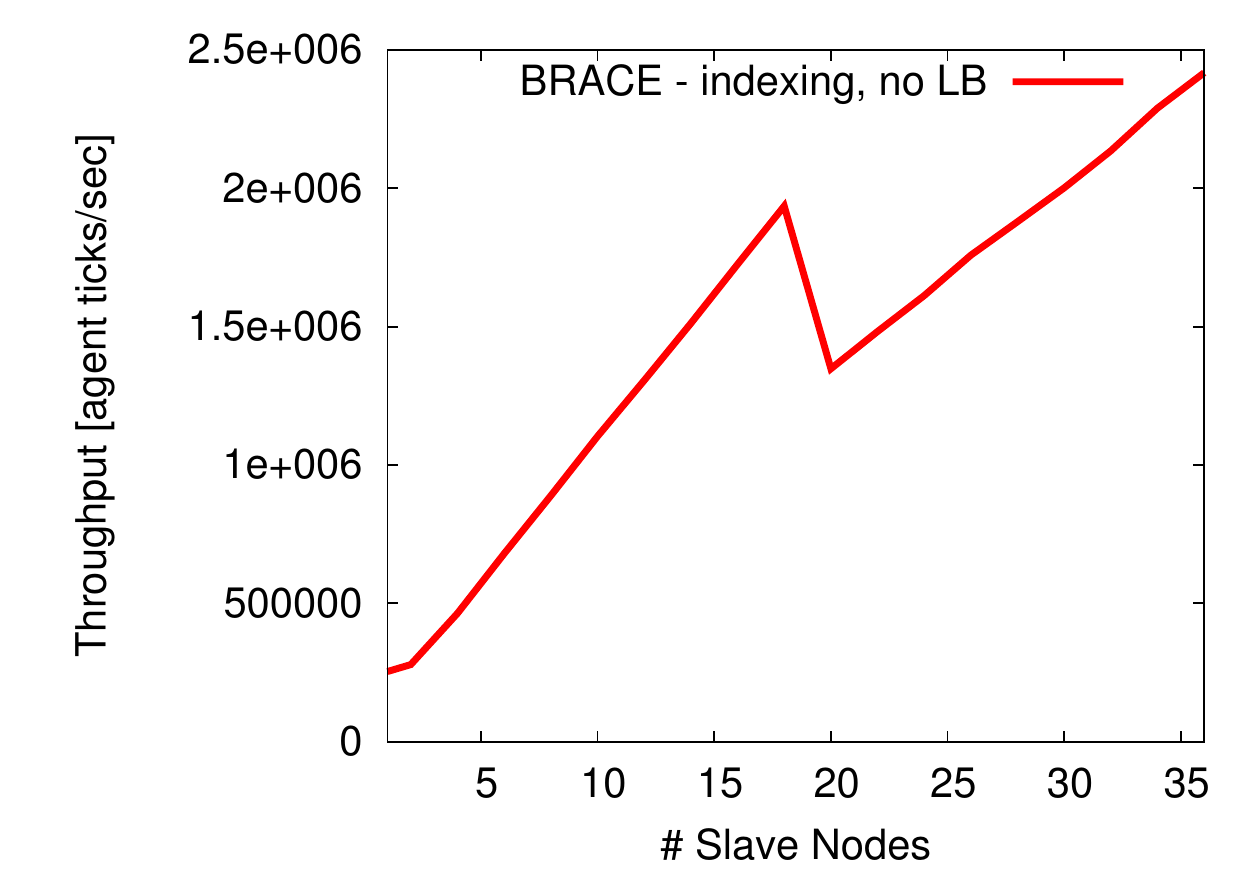}}
        \vspace{-1.5ex}
        \caption{Traffic: Scalability} \label{fig:traffic:scalability}
    \end{minipage} \hfill
        \begin{minipage}{0.327\linewidth}
        \centerline{\includegraphics[width=0.9\linewidth]{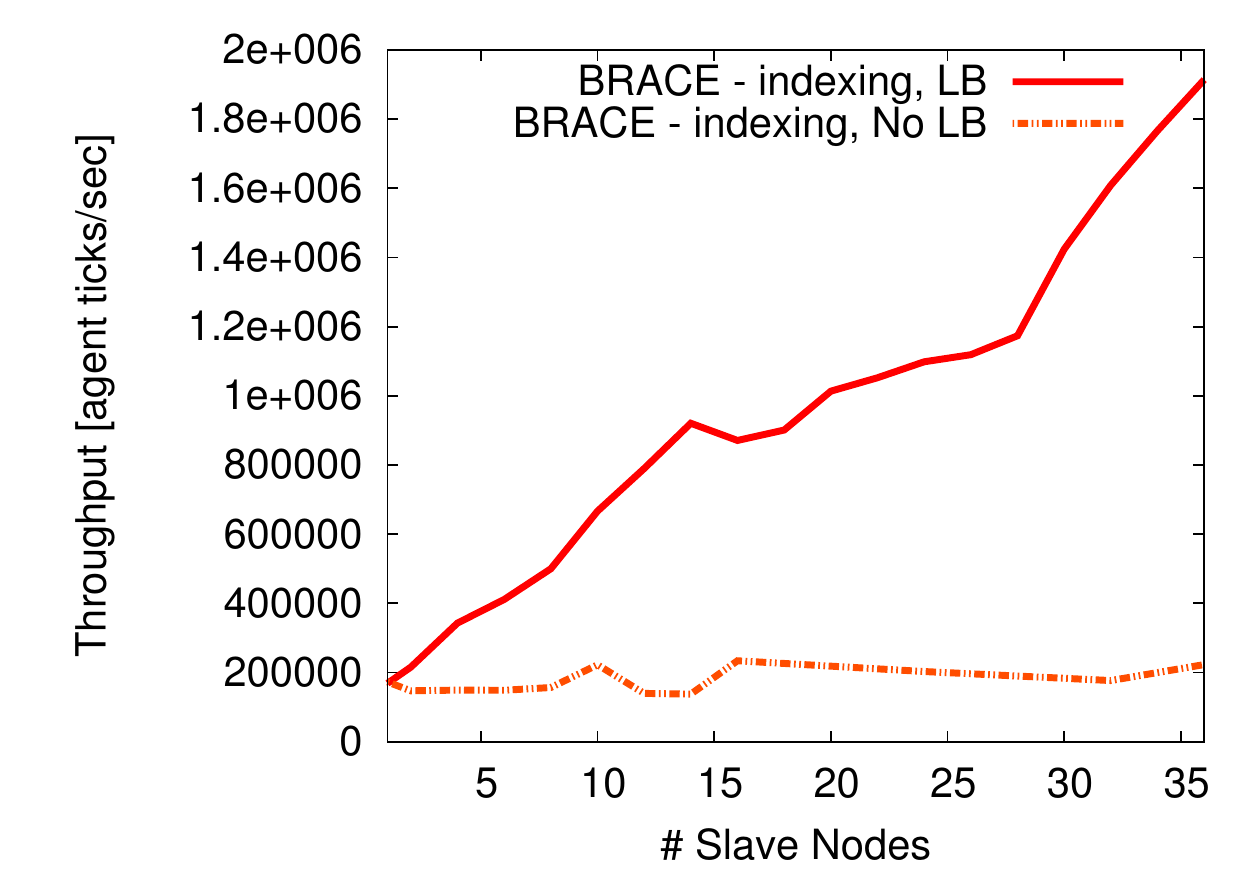}}
        \vspace{-1.5ex}
        \caption{Fish: Scalability} \label{fig:fish:scalability}
    \end{minipage} \hfill
    \begin{minipage}{0.327\linewidth}
        \centerline{\includegraphics[width=0.9\linewidth]{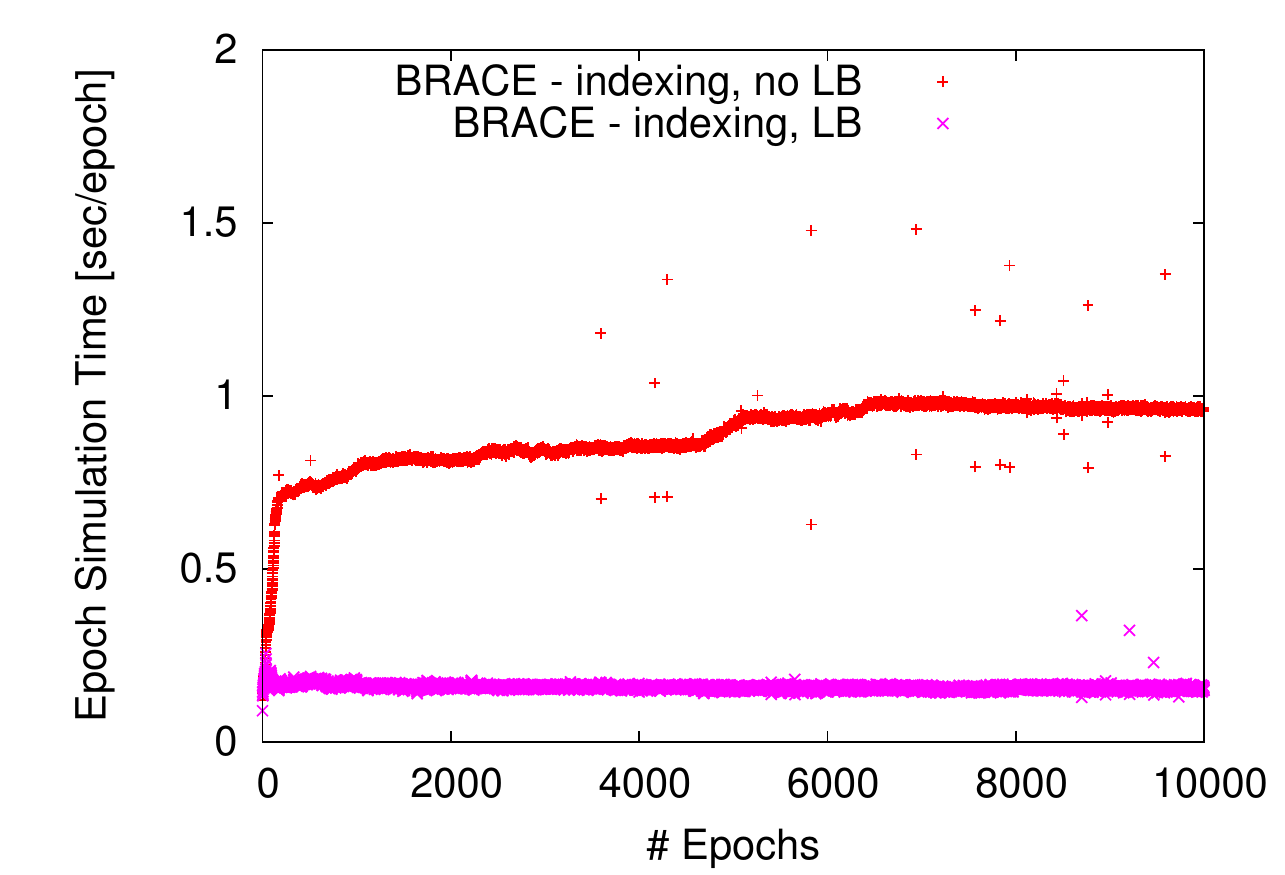}}
        \vspace{-1.5ex}
        \caption{Fish: Load Balancing} \label{fig:fish:load:balancing}
    \end{minipage} \hfill
    \vspace{-3ex}
\end{figure*}

We first compare the single-node performance of our traffic simulation to MITSIM.
The main optimization in this case is spatial indexing.
For a meaningful comparison, we validate the aggregate traffic statistics
produced by our BRASIL reimplementation
against those produced by MITSIM. Details of our validation procedure
appear in Appendix~\ref{sec:exps:model:details}.

Figure~\ref{fig:traffic:segment:length} compares the performance
of MITSIM against our BRACE reimplementation of its model.
Without spatial indexing, BRACE's performance
degrades quadratically with increasing segment length.
This is expected: In this simulation, the number of
agents grows linearly with segment length; and without indexing every vehicle enumerates
and tests every other vehicle during each tick.
With spatial indexing enabled, BRACE converts this behavior to an orthogonal range query,
resulting in log-linear growth, as confirmed by Figure~\ref{fig:traffic:segment:length}.
BRACE's spatial indexing achieves performance that is comparable, but
inferior to MITSIM's hand-coded nearest-neighbor implementation.
Our optimization techniques generalize to nearest-neighbor indexing, and
adding this to BRACE is planned future work.
With this enhancement, we expect to achieve performance parity with MITSIM.

We observed similar log-linear versus quadratic performance when scaling up the number of agents
in the fish simulation in a single node. We thus omit these results.
When we increase the visibility
range, however, the performance of the KD-tree indexing
decreases, since more results are produced for each
index probe (Figure~\ref{fig:fish:visibility}). Still,
indexing yields from two to three times improvement over a range
of visibility values.

In addition to indexing, we also measure the performance gain of eliminating
non-local effect assignments through effect inversion. Only the predator simulation has non-local
effect assignments, so we report results exclusively on this model. We run two versions of the predator simulation,
one with non-local assignments and the other with non-local assignments eliminaed by effect inversion.
We run both scripts with and without KD-tree indexing enabled on 16 slave nodes, and with BRACE configured to
have two reduce passes in the first case and only a single reduce pass in the second case. Our results are displayed
in Figure~\ref{fig:predator:inversion}. Effect inversion increases agent tick throughput
from 3.59 million (Idx-Only) to 4.36 million (Idx+Inv) with KD-tree indexing enabled, and from 2.95 million (No-Opt) to 3.63 million (Inv-Only)
with KD-tree indexing disabled. This represents an improvement of more than 20\% in each case,
demonstrating the importance of this optimization.

\subsection{Scalability of the BRACE Runtime}
\label{sec:exps:scalability}

We now explore the parallel performance of BRACE's MapReduce runtime on
the traffic and fish school simulations as we scale the number of slave
nodes from 1 to 36.
The size of both simulations is scaled linearly with the number of slaves, so we measure
scale-up rather than speed-up.

The traffic simulation represents a linear road segment with constant
up-stream traffic.  As a result, the distribution on the segment is
nearly uniform, and load is always balanced among the nodes.
Therefore, throughput
grows linearly with the number of nodes even if load
balancing is disabled (Figure~\ref{fig:traffic:scalability}).
The sudden drop around 20 nodes is an artifact of the
multi-switch architecture of the Web Lab cluster on which
we ran our experiments: Performance degrades when not all nodes
can be chosen on the same switch.

In the fish simulation, fish move in
schools led by informed individuals~\cite{CKFL05}.
In our experiment, there are two classes of informed individuals,
trying to move in two different fixed directions.
The spatial distribution of fish, and consequently the load on each slave node,
changes over time. Figure~\ref{fig:fish:scalability} shows the scalability
of this simulation with and without load balancing.
Without load balancing,
two fish schools eventually form in
nodes at the extremes of simulated space, while the load at all other
nodes falls to zero.
With load balancing,
partition grids are adjusted periodically to assign
roughly the same number of fish to each node,
so throughput increases linearly with the number of nodes.

Figure~\ref{fig:fish:load:balancing} confirms this.
With load balancing enabled, the time per simulation epoch is essentially flat;
with load balancing disabled, the epoch time gradually increases
to a value that reflects all agents being simulated by only two nodes.

\section{Related Work}
\label{sec:related}

Much of the existing work on behavioral simulations has focused on
task-parallel discrete event simulation
systems~\cite{CM78,Nic93,DFP+94,Mat93,ZKK04}. Such systems employ
either conservative or optimistic protocols to detect conflicts and
preempt or rollback simulation tasks. The strength of local
interactions and the time-stepped model used in behavioral
simulations lead to unsatisfactory performance, as shown in attempts
to adapt discrete event simulators to agent-based
simulations~\cite{HKX+06,HML04}.

Platforms specifically targeted at agent-based models have been developed,
such as Swarm~\cite{MBLA96}, Mason~\cite{LCP+05},
and Player/Stage~\cite{GVH03}.
These platforms offer tools to facilitate simulation programming,
but most rely on message-passing abstractions with implementations
inspired by discrete event simulators,
so they suffer in terms of performance and scalability.
A few recent systems attempt to distribute agent-based simulations
over multiple nodes without exploiting applications properties such as
visibility and time-stepping~\cite{HBD08,SSCD06}.
This leads either to poor scale-up
or to unrealistic restrictions on agent interactions.

Regarding join processing with MapReduce, Zhang et
al.~\cite{Zhang09:CLUSTER} compute spatial joins by an approach
similar to ours when only local effect assignments are allowed.
Their mapper partitions are derived using spatial index techniques
rather than by reasoning about the application program, and they do
not discuss {\em iterated} joins, an important consideration for our
work. Locality optimizations have been studied for MapReduce on
SMPs~\cite{YooRK09:phoenix2} and for
MapReduceMerge~\cite{YangDHRP07:MapReduceMerge}; in this paper we
consider the problem in a distributed main memory MapReduce runtime.

Data-driven parallelization techniques have also been studied in parallel databases~\cite{Dewitt90:Gamma,Graefe90:volcano} 
and data parallel programming languages~\cite{Blelloch96:nesl,SchwartzDSD86:setl}. 
However, it is unnatural and inefficient to use either SQL or set-operations
exclusively to express flexible computation over individuals
as required for behavioral simulations.

Given this situation, behavioral simulation developers have resorted to hand-coding parallel implementations of specific simulation models~\cite{EFSC09,NR01}, 
or trading model accuracy for scalability and ease of implementation~\cite{CBN03,Wen08}. To the best of our knowledge, our approach is the first to bring both programmability and scalability through data parallelism to behavioral simulations.

\section{Conclusions}
\label{sec:conclusions}

In this paper we show how MapReduce can be used to scale behavioral
simulations across clusters by abstracting these simulations as
iterated spatial joins. To efficiently distribute these joins we
leverage several properties of behavioral simulations to get a
shared-nothing, in-memory MapReduce framework called BRACE, which
exploits collocation of mappers and reducers to bound communication
overhead. In addition, we present a new scripting language for our
framework called BRASIL, which hides all the complexities of
modeling computations in MapReduce and parallel programming from
domain scientists. BRASIL scripts can be compiled into our MapReduce
framework and allow for algebraic optimizations in mappers and
reducers. We perform an experimental evaluation with two real-world
behavioral simulations to show that BRACE has nearly linear
scalability as well as single-node performance comparable to a
hand-coded simulator.

{\bf Acknowledgments.} This material is based upon work supported
by the New York State Foundation for Science, Technology, and
Innovation under Agreement C050061, by the National Science Foundation
under Grants 0725260 and 0534404, by the iAd Project funded by the
Research Council of Norway, by the AFOSR under Award 
FA9550-10-1-0202, and by Microsoft. Any opinions, findings and
conclusions or recommendations expressed in this material are those of
the authors and do not necessarily reflect the views of the funding
agencies.

\small
\bibliographystyle{abbrv}
\bibliography{main.bbl}

\balancecolumns
\normalsize

\newpage
\appendix
\section{Spatial Joins in MapReduce}
\label{sec:joins:mapreduce}

In this appendix, we formally develop the map and reduce functions for
processing a single tick of a behavioral simulation. 

\minisec{Formalizing Agents and Spatial Partitioning} We first introduce our
notation for agents and their state and effect attributes.  We denote an agent
$a$ as $a=\la\oid,\s,\e\ra$, where $\s$ is a vector of the agent's state
attributes and $\e$ is a vector of its effects. To refer to an agent or its
attributes at a tick $t$, we will write $a^t$, $\s^t$, or $\e^t$. Since effect
attributes are aggregated using combinator functions, they need to be reset at
the end of every tick. We will use $\theta$ to refer to the vector of
idempotent values for each effect. Finally, we use $\oplus$ to denote the
aggregate operator that combines effect variables according to the appropriate
combinator.

The neighborhood property implies that some subset of each agent's state
attributes are spatial attributes that determine an agent's position. For an
agent $a=\la\oid,\s,\e\ra$, we denote this spatial location $\loc(\s) \in
\rdom$, where $\rdom$ is the spatial domain. Given an agent $a$ at location
$l$, the visible region of $a$ is $\VR(l)\subseteq\rdom$.

Both the map and reduce tasks in our framework will have access to a spatial
partitioning function $\Part:\rdom\to\Pclass$, where $\Pclass$ is a set of
partition ids. This partitioning function can be implemented in multiple ways,
such as a regular grid or a quadtree.
We define the \emph{owned} set of a partition $p$ as the inverse image of $p$
under $\Part$, i.e., the set of all locations assigned to $p$. Since each
location has an associated visible region, we can also define the visible
region of a partition as $\VR(p) = \bigcup_{l\in\rdom,\Part(l)=p}\VR(l)$. This
is the set of all locations that might be visible by some agent in $p$. 

\minisec{Simulations with Local Effects Only}
Since
the query phase of an agent can only depend on the agents inside its visible
region, the visible region of a partition contains all of the data necessary
to execute the query phase for its entire owned region. We will take advantage
of this by \emph{replicating} all of the agents in this region at $p$ so that
the query phase can be executed without communication.

Figure~\ref{fig:local} shows the map and reduce functions for processing tick
$t$ when there are only local effect assignments. At tick $t$, the map
function performs the update phase from the previous tick, and the reduce
function performs the query phase. The map function takes as input an agent
with state and effect variables from the previous tick ($a^{t-1}$), and
updates the state variables to $\s^{t}$ and the effect attributes to
$\theta$. During the very first tick of the simulation, $\e^{t-1}$ is
undefined, so $\s^t$ will be set to a value reflecting the initial simulation state. The map
function emits a copy of the updated agent keyed by partition for each
partition containing the agent in its visible set
($\loc(\s^t)\in\VR(p)$). This has the effect of replicating the agent $a$ to
every partition that might need it for query phase processing.
The amount of replication depends on the partitioning function and on the
size of each agent's visible region. 

The reduce function receives as input all agents that are sent to a particular
partition $p$. This includes the agents in $p$'s owned region, as well as
replicas of all the agents that fall in $p$'s visible region. The reducer will
execute the query phase and compute effect variables for all of the agents in
its owned region (agent $i$ s.t. $\Part(\loc(\s_i^t))=p$). This requires no
communication, since the query phase of an agent in $p$'s owned region can
only depend on the agents in $p$'s visible region, all of which are replicated
at the same reducer. The reducer outputs agents with updated effect attributes
to be processed in the next tick.

\minisec{Simulations with Non-Local Effects} The method above only works when
all effect assignments are local. If an agent $a$ makes an effect assignment
to some agent $b$ in its visible region, then it must communicate that effect
to the reducer responsible for processing $b$. Figure~\ref{fig:nonlocal} shows
the complete map and reduce functions to handle simulations with non-local
effect assignments. The first map function task is same as in the local effect
case. Each agent is partitioned and replicated as necessary. As before, the
first reduce function computes the query phase for the agents in $p$'s owned
set and computes effect values. In this case, however, it can only compute
intermediate effect values $\f^t$, since it does not have the effects assigned
at other nodes. This first reducer outputs one pair for every agent, including
replicas, that has its effects updated. These agents are keyed with the
partition that owns them, so that all replicas of the same agent go to the
same node.

The second map function is the identity, and the second reduce function
performs the aggregation necessary to complete the query phase. It receives
all of the updated replicas of all of the agents in its owned region and
applies the $\oplus$ operation to compute the final effect values and complete
the query phase. Each reducer will output an updated copy of each agent in its
owned set.

\begin{figure}
  \fbox{
    \begin{minipage}{\linewidth}
      \small
      \begin{align*}
        & \mrmap^t(\cdot,a^{t-1}) = \left[(p,\la\oid,\s^t,\theta\ra) \textrm{~where~}\loc(\s^t)\in\VR(p)\right]\\
        & \mrreduce^t(p, \left[\la\oid_i,\s_i^t,\theta\ra\right]) = \left[\left(p,\la\oid_i,\s_i^t,\e_i^t\ra\right), \forall i\textrm{~s.t.~} \Part(\loc(s_i^t))=p\right]\\
      \end{align*}
      \vspace{-10ex}
      \caption{Map and reduce functions with local effects only.}
      \label{fig:local}
    \end{minipage}
  }\\
  \fbox{
    \begin{minipage}{\linewidth}
      \small
      \begin{align*}
        & \mrmap_1^t(\cdot,a^{t-1}) = \left[(p,\la\oid,\s^t,\theta\ra) \textrm{~where~}\loc(\s^t)\in\VR(p)\right]\\
        & \mrreduce_1^t(p, \left[\la\oid_i,\s_i^t,\theta\ra\right]) = \left[\left(\Part(\loc(\s^t)),\la\oid_i,\s_i^t,\f_i^t\ra\right), \forall i\textrm{~s.t.~} \f_i^t\neq\theta\right] \\
        & \mrmap_2^t(k,a) = (k,a) \\
        & \mrreduce_2^t(p, \left[\la\oid_i,\s_i^t,\f_i^t\ra\right]) = \left[\left(p,\la\oid_i,\s_i^t,\oplus_j\f_j^t\ra\right), \forall j\textrm{~s.t.~} \oid_i=\oid_j\right]
      \end{align*}
      \vspace{-6ex}
      \caption{Map and reduce functions with non-local effects.}
      \label{fig:nonlocal}
    \end{minipage}
  }
\vspace*{-4ex}
\end{figure}

\newcommand{\eval}[1]{[\![#1]\!]}

\section{Formal Semantics of BRASIL}
\label{sec:monad:overview}

In this section, we provide a more formal presentation of the
semantics of BRASIL than the one presented in
Section~\ref{sec:programming}. In particular, we show how to convert
BRASIL expressions into monad algebra expressions for analysis and
optimization.  We also prove several results regarding effect
inversion, introduced in Section~\ref{sec:programming:opt}, and
illustrate the resulting trade-offs between computation and
communication.

For the most part, our work will be in the traditional monad algebra.
We refer the reader to the original work on this algebra~\cite{BNTW1995,Koch2005,Paradaens88,TBW1992}
for its basic operators and nested data model. We also
use standard definitions for the derived operations like
cartesian product and nesting.  For example, we define
cartesian product as
\vspace{-1ex}
\begin{equation}\label{E:Cartesian}
f \times g := \la 1:f,2:g\ra \circ \pwith{1} \circ \flatmap(\pwith{2})
\vspace{-1ex}
\end{equation}
For the purpose of readability, composition in
\eqref{E:Cartesian} and the rest of our presentation,
is read left-to-right; that is, $(f \circ g)(x) = g\big(f(x)\big)$.

We assume that the underlying domain is the real numbers, and that we
have associated arithmetic operators. We also add traditional aggregate
functions like $\textsc{count}$ and $\textsc{sum}$ to the algebra; these
functions take a set of elements (of the appropriate type) and return a
value.

In order to simplify our presentation, we do make several small changes
that relax the typing constraints in the classic monad algebra.  In
particular, we want to allow union to combine sets of tuples with
different types.  For this end, we introduce a special $\nil$ value.
This value is the result of any query that is undefined on the input
data, such as projection on a nonexistent attribute.  This value has
a form of  ``null-semantics'' in that values combined with $\nil$ are
$\nil$, and $\nil$ elements in a set are ignored by aggregates. In
addition, we introduce a special aggregate function $\get$.  When given
a set, this function returns its contents if it is a singleton, and
returns $\nil$ otherwise.  Neither this function, nor the presence
of $\nil$ significantly affects the expressive power of the monad
algebra~\cite{Suciu95}.

\subsection{Monad Algebra Translation}
\label{sec:monad:formal}

For the purpose of illustration, we assume that our simulation has only
one class of agents, all of which are running the same simulation script.
It is relatively easy to generalize our approach to multiple agent classes
or multiple scripts.  Given this assumption, our simulation data is
simply a set of tuples $\{ t_0, \ldots, t_n \}$ where each tuple $t_i$
represents the data inside of an agent.  Every agent tuple has a special
attribute \key which is used to uniquely identify the agent; variables
which reference another agent make use of this key.  The state-effect
pattern requires that all data types other than agents be processed by
value, so they can safely be stored inside each agent.

We let $\tau$ represent the type/schema of an agent.  In addition to the
key attribute, $\tau$ has an attribute for each state and effect field.
The value of a state attribute is the value of the field.  The value of
an effect attribute is a pair $\la 1: n, 2: \agg\ra$ where $n$ is a
unique identifier for the field and $\agg$ is the aggregate for
this effect.

During the query phase, we represent effects as a tuple
$\la k:\D{N}, e:\D{N}, v:\sigma \ra$, where $k$ is the key of the object
being effected, $e$ is the effect field identifier, and $v$ is the value
of the effect. As a shorthand, let $\rho$ be this type.  Even though effects
may have different types, because of our relaxed typing, this will
not harm our formalism.

The syntax of BRASIL forces the programmer to clearly separate the code
into a query script (i.e.\ \texttt{run()}) and an update script
(the update rules).  A query script compiles to a expression whose
input and output are the tuple $\la 1: \tau', 2: \{ \tau \}, 3: \{\rho\}\ra$.
The first element represents the active agent for this script; $\tau'$
``extends'' type $\tau$ in that it is guaranteed to have an attribute
for the key and each state field, but it may have more attributes.
The second element is the set of all other agents with which this
agent must communicate.  The last element is the set of effects
generated by this script.

Let $Q$ be the monad expression for the query script.  Then the effect
generation stage is the expression
\vspace{-1ex}
\begin{equation}
\F{Q}(Q) = (\id \times \id)\circ\nest_2\circ\map\big(\widehat{Q})
\vspace{-1ex}
\end{equation}
where $\widehat{Q}$ is defined as
\vspace{-1ex}
\begin{equation}
\widehat{Q} =
    \la\!1\!:\!\pi_1, 2\!:\!\pi_2, 3\!:\!\{\!\}\!\ra\circ Q \circ
    \la\!1\!:\!\pi_1, 2\!:\!\pi_3\ra
\vspace{-1ex}
\end{equation}
This produces a set of agents and the effects that they have generated
(which may or may not be local).  In general, we will aggregate aggressively,
so each agent will only have one effect for each pair $k,e$.  For the
effect aggregation stage, we must aggregate the effects for each agent and
inline them into the agent tuple.  If we only have local effect assignments,
then this expression is $\F{Q}(Q) \circ \F{E}$ where
\vspace{-1ex}
\begin{equation}
\F{E} =
    \map\big(
        \la\!\key\!:\!\pi_\key, s_i\!:\!\pi_{s_i}, e_j\!:\!%
            \pi_2\circ\sigma_{\pi_e = \pi_{e_j}\circ\pi_1}\circ(\pi_{e_j}\circ\pi_2)\!%
            \ra_{i,j}
    \big)
\vspace{-1ex}
\end{equation}
where the $s_i$ are the state fields and the $e_j$ are the effect
fields. However, in the case where we have non-local effects, we
must first redistribute them via the expression \vspace{-1ex}
\begin{equation}
\F{R} =
    (\pi_1 \times \pi_2) \circ
    \map\big(
        \la\!1\!:\!\pi_1,2\!:\!\flatten\circ\sigma_{\pi_k = \pi_1\circ\pi_\key}\ra
    \big)
\vspace{-1ex}
\end{equation}
So the entire query phase is $\F{Q}(Q)\circ\F{R}\circ\F{E}$.  Finally, for
the update phase, each state $s_i$ has an update rule which corresponds to an
expression $U_{s_i}$.  These scripts read the output of the expression $\F{E}$.
Hence the query for our entire simulation is the expression
\vspace{-1ex}
\begin{equation}
\F{Q}(Q)\circ\F{R}\circ\F{E}\circ\F{U}(U_{s_0},\ldots,U_{s_n})
\vspace{-1ex}
\end{equation}
where the update phase is defined as
\vspace{-1ex}
\begin{equation}
\F{U}(U_{s_0},\ldots,U_{s_n}) =
    \map\big(
        \la \key:\pi_\key, s_0:U_{s_0}, \ldots, s_n:U_{s_n}\ra
    \big)
\vspace{-1ex}
\end{equation}

The only remaining detail in our formal semantics is to define semantics
for the query scripts and update scripts.  Update scripts are just simple
calculations on a tuple, and are straightforward.  The only nontrivial
part concerns the query scripts.  A script is just a sequence of statements
$S_0;\ldots;S_n$ where each statement is a variable declaration, assignment,
or control structure (e.g.\ conditional, \texttt{foreach}-loop).
See the BRASIL Language manual for more information on the complete grammar~\cite{BRASILManual}.
It suffices to define, for each statement $S$, a monad algebra expression
$\eval{S}$ whose input and output are the triple
$\la 1: \tau', 2: \{ \tau \}, 3: \{\rho\}\ra$; we handle
sequences of commands by composing these expressions.

\begin{figure}
\begin{equation*}
\begin{split}
\eval{&\T{const $\tau$ $x$ = $E$}}_V =
\la 1\!:\!\chi_x(\eval{E}_V),2\!:\!\pi_2,3\!:\!\pi_3\ra \\
\eval{&\T{effect $\tau$ $x$ :\ $f$}}_V =
\la 1\!:\!\chi_x(\la 1:\rho(x),2:f\ra),2\!:\!\pi_2,3\!:\!\pi_3\ra \\
\eval{&\T{$x$ <- $E$}}_V =
    \la 1\!:\!\pi_2, 2\!:\!\pi_2, 3\!:\!\pi_3\oplus \\
&\phantom{\T{$x$ <- $E$}]\!]_V\> = \la 1:}
        (\la 1\!:\!\pi_1\circ\pi_\key, 2\!:\!\rho(x), 3\!:\!\eval{E}_V\ra\circ\sng)\ra \\
\eval{&\T{$R.x$ <- $E$}} =
    \la 1\!:\!\pi_2, 2\!:\!\pi_2, 3\!:\!\pi_3\oplus\\
&\phantom{\T{$R.x$ <- $E$}]\!]_V\> = \la}
     (\la 1\!:\!\eval{R}_v, 2\!:\!\rho(x), 3\!:\!\eval{E}_V\ra\circ\sng)\ra \\
\eval{&\T{if ($E$) \{$B_1$\} else \{$B_2$\}}}_V = \\
&\qquad\qquad
    \la 1:\pi_2, 2:\pi_2,
        3:\sng\circ\sigma_{\eval{E}_V}\circ\get\circ \eval{B_1}_V \oplus \\
&\phantom{\qquad\qquad\la 1:\pi_2, 2:\pi_2,
    3:} \sng\circ\sigma_{\neg \eval{E}_V}\circ\get\circ \eval{B_2}_V
\ra \\
\eval{&\T{foreach ($\tau$ x :\ $E$) \{$B$\}}}_V = \\
&\qquad\qquad\la 1\!:\!\pi_2, 2\!:\!\pi_2, \\
&\phantom{\qquad\qquad\la}
    3\!:\!\la 1\!:\!\pi_1\circ\chi_x(\eval{E}_V)\circ\pwith{x},2\!:\!\pi_2, 3\!:\!\pi_3\ra \\
&\phantom{\qquad\qquad\la 3\!:}
    \circ \flatmap(\eval{B}_V \circ \pi_3)
    \ra
\end{split}
\end{equation*}
\vspace{-3ex}
\caption{Translation for Common Commands}
\label{Fi:Commands}
\vspace{-2ex}
\end{figure}

Recall that our query script semantics depends on the visibility
constraints in the script.  We generalize the approach from
Section~\ref{sec:programming:brasil} and represent visibility as a
predicate $V(x,y)$ which compares two agents.  For any statement $S$, we
let $\eval{S}$ be its interpretation with this constraint and
$\eval{S}_V$ be the semantics without.

Before translating statements, we must translate expressions that may
appear inside of them.  The only nontrivial expressions are references;
arithmetic expressions or other complex expressions translate to the
monad algebra in the natural way. References return either the variable
value, or the key for the agent referenced. Ignoring visibility constraints,
for any identifier $x$, we define
\vspace{-1ex}
\begin{equation}\label{E:Standard:Reference}
\eval{x} =
\begin{cases}
\begin{split}
\pwith{3}&\circ\sigma_{\pi_1\circ\pi_x\circ\pi_1 = \pi_3\circ\pi_e} \\
&\circ\sigma_{\pi_1\circ\pi_\key = \pi_3\circ\pi_k}\circ\get
\end{split}
&\text{$E$ is \texttt{effect}} \\
\pi_1\circ\pi_x &\text{otherwise}
\end{cases}
\vspace{-1ex}
\end{equation}
In general, for any reference $E.x$, we define
\vspace{-1ex}
\begin{equation}\label{E:Complex:Reference}
\eval{E.x} = \la 1:\pi_2\circ\sigma_{\pi_\key = \eval{E}}\circ\get,2:\pi_2,3:\pi_3\ra\circ\eval{x}
\vspace{-1ex}
\end{equation}
If we include visibility constraints, $\eval{E}_V$ is defined in much the
same way as $\eval{E}$ except when $E$ is an agent reference.  In that case,
\vspace{-1ex}
\begin{equation}\label{E:Visible:Reference}
\begin{split}
\eval{x}_V =
    \la\!1\!:\!\id, 2\!:\!\pi_2\circ\sigma_{\pi_\key = \eval{E}}\circ\get\!\ra&\circ
        \la\!1\!:\!V, 2\!:\!\pi_2\!\ra\circ\sng \\
        &\circ\sigma_{\pi_1}\circ\get\circ\pi_2\circ\pi_k \\
\end{split}
\vspace{-1ex}
\end{equation}
This expression temporarily retrieves the object, tests if it is visible, and
returns $\nil$ if not.

To complete our semantics, we introduce the following notation.
\vspace{-3ex}
\begin{itemize}
\item
$\chi_a(f)$ is an operation that takes a tuple and extends it with an
attribute $a$ having value $f$. It is definable in the monad algebra, but its
exact definition depends on its usage context.
\item
$\oplus$ is an operations that takes two sets of effects and aggregates those
with the same key and effect identifier.  It is definable on in the monad
algebra, but its exact definition depends on the effect fields in the
BRASIL script.
\item
$\rho(x)$ is the effect identifier for a variable $x$.  In practice,
this is the position of the declaration of $x$ in the BRASIL script.
\vspace{-1ex}
\end{itemize}
Given these three expressions, Figure~\ref{Fi:Commands} illustrates the
translation of some of the more popular statements in the monad algebra.
In general, variable declarations modify the first element of the input
triple (i.e.\ the active agent), while assignments and control structures
modify the last element (i.e.\ the effects).

As we discussed in Section~\ref{sec:programming:opt}, this formalism allows
us to apply standard algebraic rewrites from the monad algebra for
optimization.  For example, many of the operators in Figure~\ref{Fi:Commands}
-- particularly the tuple constructions -- are often unnecesary.  They
are there to preserve the input and output format, in order to facility
composition.  There are rewrite rules that function like dead-code
elimination, in that they remove tuples that are not being used.  One
of the consequences is that many \texttt{foreach}-loops simplify
to the form
\vspace{-1ex}
\begin{equation}\label{E:Simple:For}
F(E,B) = \la 1: \id, 2:E\ra \circ \pwith{2}\circ\flatmap(B)
\vspace{-1ex}
\end{equation}
Note that this form is ``half'' of the cartesian product in \eqref{E:Cartesian};
it joins a single value with a set of values.  Thus when we simplify the
\texttt{foreach}-loop to this form, we can often apply join
optimization techniques to the result.

Another advantage of this formalism is that it allows us to prove correctness
results.  For example, the semantics of the visibility constraints in BRASIL
is defined in terms of weak references.  However, our implementation involves
restricting the read set of each agent to those that are visible.  It is a
simple exercise to use our formalism to prove that theses approaches are
equivalent.  The following result is the formal version of Theorem~\ref{T:Visible}
from Section~\ref{sec:programming:brasil}.

\vspace{-1ex}
\begin{theorem}\label{T:Visible}
Let $Q$ be a BRASIL query script whose references are restricted by visibility
predicate $V$. Then
\vspace{-1ex}
\begin{equation}\label{E:Visible:Read}
\nest_2  \circ \map(\widehat{\eval{Q}_V}) =
\sigma_V \circ \nest_2 \circ \map(\widehat{\eval{Q}})
\vspace{-1ex}
\end{equation}
Furthermore, let
\vspace{-1ex}
\begin{equation}
O(F) = F\circ (\pi_2\times\pi_3) \circ \sigma_{\pi_1\circ\pi_\key = \pi_2\circ\pi_k} \circ \map(\pi_1)
\vspace{-1ex}
\end{equation}
be the set of objects affected by an expression $F$.  Then
\vspace{-1ex}
\begin{equation}\label{E:Visible:Write}
\map(\la 1\!:\!\pi_1, 2\!:\!O(\eval{Q}_V) \ra )\circ \sigma_V = \map(\la 1\!:\!\pi_1, 2\!:\!O(\eval{Q}_V) \ra )
\vspace{-1ex}
\end{equation}
\end{theorem}

The significance of \eqref{E:Visible:Read} is that, instead of implementing
the overhead of checking for weak references, we can filter out the
agents that are not visible and eliminate any further visibility checking.
The significance of \eqref{E:Visible:Write} is that weak references insure
agents can only affect visible agents.

\subsection{Effect Inversion}
\label{sec:monad:inversion}

As we saw in Section~\ref{sec:programming:opt}, there is an advantage to
writing a BRASIL script so that all effects assignments are local.  It may
not always be natural to do so, as the underlying scientific models may be
expressed in terms of non-local effects.  However, in certain cases, we
may be able to automatically rewrite a BRASIL program to only use local
effects.  In particular, if there are no visibility constraints, then we
can always invert effect assignments to make them local-only.  The following
is the formal version of the result stated in Section~\ref{sec:programming:opt}.

\vspace{-1ex}
\begin{theorem}\label{T:Invert}
Let $Q$ be a query script with no visibility constraints.  There is a
script $Q'$ with only local effects such that $\eval{Q} = \eval{Q'}$.
\end{theorem}

\begin{proof}[Sketch]
Our proof takes advantage of the fact that effect fields (as opposed to
effect variables) may not be read during the query phase, and that
effects are aggregated independent of order.  We start with $Q$ and
create a copy script $Q_1$.  Within this copy, we remove all syntactically
non-local effect assignments (e.g.\ $E.x \T{ <- } v$).  Some of these may
actually be local in the semantic sense, but this does not effect our proof.

We construct another copy $Q_2$.  For this copy, we pick a variable $a$ that
does not appear in $Q$.  We replace every local state reference $x$ in $Q$
with $a.x$. We also remove all local effect assignments.  Finally, we replace
each syntactically non-local assignment $E.x \T{ <- } v$ with the conditional
assignment \T{if ($E$ == this) \{$x$ <- $v$\}}. We then let $Q_3$ be the
script
\vspace{-1ex}
\begin{equation*}
\T{foreach(Agent a :\ Extent<Agent>) \{ $Q_2$(a) \}}
\vspace{-1ex}
\end{equation*}
That is, $Q_3$ is the act of an agent running the script for each other
agent, searching for effects to itself, and then assigning them locally.
The script $Q_1;Q_3$ is our desired script.
\end{proof}

Note that this conversion comes at the cost of an additional
\texttt{foreach}-loop, as each agent simulates the actions of all other
agents.  Thus, this conversion is much more computationally
expensive than the original script.  However, we can often simplify
this to remove the extra loop.  As mentioned previously, a
\texttt{foreach}-loop can often be simplified to the form in
\eqref{E:Simple:For}.  In the case of two nested loops over the same
set $E$, the merging of these two loops is a type of self-join.
That is,
\vspace{-1ex}
\begin{equation*}
\begin{split}
F(E,\!F(\!E,\!B\!)) &\!=\!
    \la\!1\!:\!\id, 2\!:\!E\!\ra \circ \pwith{2}\circ \\
&\phantom{=\la}
    \flatmap\big(
        \la\!1\!:\!\id, 2\!:\!E\!\ra\!\circ\!\pwith{2}\!\circ\!\flatmap(B)
    \big) \\
&=
\la 1:\id, 2:E, 3:E\ra\circ\pwith{2}\circ \\
&\phantom{=\la}
    \flatmap(\pwith{3}\circ\flatmap(B'))\\
&= \la 1:\id, 2\!:\!(E\times E)\ra\circ\pwith{2}\circ\flatmap(B')
\end{split}
\vspace{-1ex}
\end{equation*}
where $B'$ and $B''$ are $B$ rewritten to account for the change in
tuple positions.  As part of this rewrite, may discover that that
self-join is redundant in the expression $B''$ and eliminate it;
this is how we get simple effect inversions like the one illustrated
in Section~\ref{sec:programming:opt}.

In the case of visibility constraints, the situation becomes a little
more complex.  In order to do the inversion that we did the proof of
Theorem~\ref{T:Invert}, we must require that any agent $a_1$ that
assigns effects to another agent $a_2$ must restrict its visibility
to agents visible to $a_2$; that way $a_2$ can get the same results
when it reproduces the actions of $a_1$.  This is fairly restrictive,
as it suggests that every agent needs to be visible to every other
agent.

We can do better by introducing an information flow analysis. We only
require that, for each non-local effect assigned to agent, that effect
is computed using only information from agents visible to the one being
assigned. However, this property depends on the values of the agents,
and cannot (in generally) be inferred statically from the script.  Thus it
is infeasible to exploit this property in general.

However, there is another way to invert scripts in the phase of
visibility constraints.  Suppose the visibility constraint for
a script $Q$ is a distance bound, such as $d(x,y) < R$.  If we
relax the visibility constraint for the script in the proof of
Theorem~\ref{T:Invert} to $d(x,y) < 2R$, then the proof carries
through again.  We state this modified result as follows:

\vspace{-1ex}
\begin{theorem}\label{T:Visible:Invert}
Let $Q$ be a query script with visibility constraint $V$.  Let $V'$
be such that $V'(x,y)$ if and only if $\exists z V(x,z) \wedge V(z,y)$.
Then there is a script $Q'$ with only local effects such that
$\eval{Q}_V = \eval{Q'}_{V'}$.
\end{theorem}

\begin{proof}[Sketch]
The proof is similar to that of Theorem~\ref{T:Invert}.  The only
difference is that we have to ensure that the increased visibility
for $Q'$ does not cause the weak references in a script to resolve
to agents that would have otherwise evaluated to $\nil$.  In the
construction of $Q_2$, we use local constants to normalize the
expressions so that any agent reference in the original script
becomes a local constant. For example, suppose each agent has a
field \texttt{friend} that is a reference to another agent.  If we
have a conditional of the form
\vspace{-1ex}
\begin{verbatim}
   if (friend.x - x < BOUND) { ... }
\end{verbatim}
\vspace{-1ex}
then we normalize this expression as
\vspace{-1ex}
\begin{verbatim}
   const agent temp = friend;
   if (temp.x - x < BOUND) { ... }
\end{verbatim}
\vspace{-1ex}
When then wrap these introduced constants with conditionals that
test for visibility with respect to the old constraints.  For example,
the code above would become
\vspace{-1ex}
\begin{verbatim}
   const agent temp = (visible(this,friend) ?
                         friend : null);
   if (temp.x - x < BOUND) { ... }
\end{verbatim}
\vspace{-1ex}
where \T{visible} is a method evaluating the visibility constraint
and $\null$ evaluates to $\nil$ in the monad algebra. Given the semantics
of $\nil$, this translation has the desired result.
\end{proof}

\section{Details of Simulation Models}
\label{sec:exps:model:details}

This section describes the two simulation models we have implemented
for BRACE single-node performance and scalability experiments.

\noindent \textbf{Traffic Simulation.} Traffic simulation is
required to provide accurate, high-resolution, and realistic
transportation activity for the design and management of
transportation systems. MITSIM, a state-of-the-art single-node
behavioral traffic simulator has several different models covering
different aspects of driver behavior~\cite{YKB99}. For example,
during each time step, a lane selection model will make the driver
inspect the lead and rear vehicles as well as the average velocity
of the vehicles in her current, left, and right lanes (within
lookahead distance parameter $\rho$) to compute the utility function
for each lane. A probabilistic decision of lane selection is then
made according to the lane utility. If the driver decides to change
her lane, she needs to inspect the gaps from herself to the lead and
rear vehicles in the target lane to decide if it is safe to change
to the target lane in the next time step. Otherwise, the vehicle
following model is used to adapt her velocity based on the lead
vehicle. The newly computed velocity will replace the old velocity
in the next time step. Note that if the driver cannot find a lead or
rear vehicle within $\rho$, she will just assume the distance to the
lead or rear vehicle is infinite, and adjust the velocity according
to a free-flow submodel. Because only limited information about MITSIM's
driving models is available in the literature~\cite{YKB99}, we found
it crucial to ensure that our implementation of MITSIM's lane
changing and acceleration models was as accurate as possible.

Therefore we validate consistency of the MITSIM model encoded in
BRASIL in terms of the simulated traffic conditions. One note is
that since the MITSIM model hand-coded a nearest neighbor indexing
for accessing the lead and rear vehicles, its lookahead distance
actually varies for each vehicle. In our reimplementation we fix
this distance to 200 in order to apply single-node spatial indexing.
We compare lane changing frequencies, average lane velocity and
average lane density with the segment length 20,000 on both
simulators. The statistical difference is measured by RMSPE
(Relative Mean Square Percentage Error), which is often used as a
goodness-of-fit measure in the traffic simulation
literature~\cite{CBTRL06}. The results for all these three
statistics are shown in Table~\ref{table:traffic:accuracy}. We can
see that except for Lane 4's average density and changing frequency,
all the other statistics demonstrate strong agreement between the
two simulators. This exception is due to the fact that in the MITSIM
lane changing model drivers have a reluctance factor to change to
the right most lane (i.e., Lane 4). As a result there are only a few
vehicles on that lane (56.33 vehicles on average compared to 351.42
on other lanes), and small lane changing record deviations due to
the fixed lookahead distance approximation can contribute
significantly to the error measurement.

\begin{table}
  \begin{tabular}{| l | c | c | c |}
  \hline
  \textbf{Lane} & \textbf{Change Frequency} & \textbf{Avg. Density} & \textbf{Avg. Velocity} \\
  \hline
  L1 & 8.93\% & 7.42\% & 0.007\% \\
  \hline
  L2 & 5.57\% & 10.38\% & 0.007\% \\
  \hline
  L3 & 7.67\% & 9.38\% & 0.007\% \\
  \hline
  L4 & 21.37\% & 19.72\% & 0.007\% \\
  \hline
  \end{tabular}
\caption{RMSPE for Traffic Simulation (LookAhead = 200)}
\label{table:traffic:accuracy} \vspace*{-1ex}
\end{table}

\noindent \textbf{Fish School Simulation.}
Couzin et al. have built a behavioral fish school simulation model to study
information transfer in groups of fish when group members cannot recognize
which companions are informed individuals who know about a food
source~\cite{CKFL05}. This computational model proceeds in time steps, i.e.,
at each time period each fish inspects its environment to decide on the
direction which it will take during the next time period. Two basic behaviors
of a single fish are avoidance and attraction. Avoidance has the higher
priority: Whenever a fish is too close to others (i.e., distance less than a
parameter $\alpha$), it tries to turn away from them. If there is no other
fish within distance $\alpha$, then the fish will be attracted to other fish
within distance $\rho > \alpha$. The influence will be summed and normalized
with respect to the current fish. Therefore, any other individuals out of the
visibility range $\rho$ of the current individual will not influence its
movement decision. In addition, informed individuals have a preferred
direction, e.g., the direction to the food source or the direction of
migration. These individuals will balance the strength of their social
interactions (attraction and avoidance) with their preferred direction
according to a weight parameter $\omega$.

\noindent \textbf{Predator Simulation.} Since both the traffic and
the fish school simulations only use local effect assignments, we
designed a new predator simulation, inspired by artificial society
simulations~\cite{EA96}. In this simulation, a fish can ``spawn''
new fish and ``bite'' other fish, possibly killing them, so density
naturally approaches an equilibrium value at which births and deaths
are balanced. Since effect inversion is not yet implemented in the
BRASIL Compiler, we program biting behavior either as a non-local
effect assignment (fish assign ``hurt'' effects to others) or as a
local one (fish collect ``hurt'' effects from others) in otherwise
identical BRASIL scripts.

\end{document}